\documentclass[11pt]{article}

\usepackage{hyperref}

\usepackage{amsfonts,amsmath,amssymb}

\usepackage{enumerate, graphicx}
\usepackage{xspace}
\usepackage{boxedminipage}
\usepackage{url}
\usepackage{float}
\usepackage{enumerate,paralist}
\usepackage{accents}
\usepackage{setspace}
\usepackage{mathrsfs}
\usepackage{paralist}%
\usepackage{scalerel}%
\usepackage{caption}%
\usepackage[in]{fullpage}

\usepackage[amsmath,thmmarks]{ntheorem}%
\theoremseparator{.}%

\usepackage{color}

\newcommand{\Do}{{\small\bf do}\ }

\newcommand{\Return}{{\small\bf return\ }}

\newcommand{\If}{{\small\bf if}\ }
\newcommand{\Then}{{\small\bf then}\ }
\newcommand{\Else}{{\small\bf else}\ }

\newcommand{\For}{{\small\bf for}\ }

\newcommand{\xbeginlgox}{\begin{minipage}{1in}\begin{tabbing}
           \quad\=\qquad\=\qquad\=\qquad\=\qquad\=\qquad\=\qquad\=\kill}
        \newcommand{\xendlgox}{\end{tabbing}\end{minipage}}
\newcommand{\xbeginlgoxx}{\begin{minipage}{1in}\begin{tabbing}
           \=\quad\=\quad\=\quad\=\quad\=\quad\=\quad\=\kill}
        \newcommand{\xendlgoxx}{\end{tabbing}\end{minipage}}

\newenvironment{algorithmEnv}{
   \begin{tabular}{|l|}\hline\xbeginlgox\bigskip}
    {\xendlgox\\\hline\end{tabular}}
\newenvironment{algorithmEnvX}{
   \begin{tabular}{|l|}\hline\xbeginlgoxx\bigskip}
    {\xendlgoxx\\\hline\end{tabular}}

\newenvironment{program}{
   \begin{minipage}{4.0in}
   \begin{tabbing}
       \ \ \ \ \= \ \ \ \= \ \ \ \ \= \ \ \ \ \= \ \ \ \ \=
      \ \ \ \ \= \ \ \ \ \= \ \ \ \ \= \ \ \ \ \=
      \ \ \ \ \= \ \ \ \ \= \ \ \ \ \= \ \ \ \ \= \kill
}{
   \end{tabbing}
   \end{minipage}
}

\DefineNamedColor{named}{RedViolet} {cmyk}{0.07,0.90,0,0.34}

\newcommand{\AlgorithmI}[1]{{\textcolor[named]{RedViolet}{\texttt{\bf{#1}}}}}
\newcommand{\Algorithm}[1]{{\AlgorithmI{#1}\index{algorithm!#1@{\AlgorithmI{#1}}}}}

\newcommand{\CodeComment}[1]{\textcolor{blue}{\texttt{#1}}}

\usepackage{mleftright}%

\usepackage{hyperref}%
\hypersetup{%
   breaklinks,%
   ocgcolorlinks,
   colorlinks=true,%
   linkcolor=[rgb]{0.45,0.0,0.0},%
   citecolor=[rgb]{0,0,0.45}
}

\numberwithin{figure}{section}
\numberwithin{table}{section}
\numberwithin{equation}{section}%

\renewcommand{\th}{th\xspace}

\newcommand{\HLinkShort}[2]{\hyperref[#2]{#1\ref*{#2}}}
\newcommand{\HLink}[2]{\hyperref[#2]{#1~\ref*{#2}}}
\newcommand{\HLinkPage}[2]{\hyperref[#2]{#1~\ref*{#2}%
      $_\text{p\pageref{#2}}$}}
\newcommand{\HLinkPageOnly}[1]{\hyperref[#1]{Page~\refpage*{#1}%
      $_\text{p\pageref{#1}}$}}

\newcommand{\HLinkSuffix}[3]{\hyperref[#2]{#1\ref*{#2}{#3}}}
\newcommand{\HLinkPageSuffix}[3]{\hyperref[#2]{#1\ref*{#2}%
      #3$_\text{p\pageref{#2}}$}}

\newcommand{\figlab}[1]{\label{fig:#1}}
\newcommand{\figref}[1]{\HLink{Figure}{fig:#1}}

\newcommand{\seclab}[1]{\label{sec:#1}}
\newcommand{\secref}[1]{\HLink{Section}{sec:#1}}

\newcommand{\corlab}[1]{\label{cor:#1}}
\newcommand{\corref}[1]{\HLink{Corollary}{cor:#1}}%

\providecommand{\deflab}[1]{\label{def:#1}}
\newcommand{\defref}[1]{\HLink{Definition}{def:#1}}

\newcommand{\lemlab}[1]{\label{lemma:#1}}
\newcommand{\lemref}[1]{\HLink{Lemma}{lemma:#1}}%

\newcommand{\obslab}[1]{\label{observation:#1}}
\newcommand{\obsref}[1]{\HLink{Observation}{observation:#1}}

\newcommand{\thmlab}[1]{{\label{theo:#1}}}
\newcommand{\thmref}[1]{\HLink{Theorem}{theo:#1}}

\providecommand{\eqlab}[1]{}%
\renewcommand{\eqlab}[1]{\label{equation:#1}}

\theoremstyle{plain}%
\newtheorem{theorem}{Theorem}[section]

\newtheorem{lemma}[theorem]{Lemma}

\newtheorem{corollary}[theorem]{Corollary}
 
\newtheorem{observation}[theorem]{Observation}

\theoremstyle{plain}%
\theoremheaderfont{\sf} \theorembodyfont{\upshape}%
\newtheorem*{remark:unnumbered}[theorem]{Remark}%
\newtheorem{remark}[theorem]{Remark}%
\newtheorem{definition}[theorem]{Definition}

\theoremheaderfont{\em}%
\theorembodyfont{\upshape}%
\theoremstyle{nonumberplain}%
\theoremseparator{}%
\theoremsymbol{\myqedsymbol}%
\newtheorem{proof}{Proof:}%

\newcommand{\myqedsymbol}{\rule{2mm}{2mm}}

\newtheorem{proofof}{Proof of\!}%

\def\compactify{\itemsep=0pt \topsep=0pt \partopsep=0pt \parsep=0pt}
\let\latexusecounter=\usecounter

{\begin{itemize}%
\setlength{\itemsep}{0pt}%
\setlength{\topsep}{0pt}%
\setlength{\partopsep}{0pt}%
\setlength{\parskip}{0pt}}%
{\end{itemize}}

\newcommand{\eps}{\varepsilon}%
\def\bar{\overline}

\def\script#1{\mathcal{#1}}
\def\opt{\textsc{OPT}}

\def\etal{\text{et al.}\xspace}

\def\sep{\;|\;}
\def\card#1{|#1|}
\def\set#1{\{#1\}}

\newcommand{\sol}{\ensuremath{\mathtt{sol}}}%
\newcommand{\solS}{\mathtt{sol}^{}_\GSample}%
\def\path{\mathrm{path}}

\newcommand{\polylog}{\mathop{\mathrm{polylog}}}%
\newcommand{\probSetCover}{\prob{Set\,Cover}{}\xspace}

\newcommand{\pbrcx}[1]{\left[ {#1} \right]}%
\newcommand{\Prob}[1]{\mathop{\mathbf{Pr}}\!\pbrcx{#1}}

\newcommand{\pth}[1]{\mleft({#1}\mright)}%
\newcommand{\brc}[1]{\left\{ {#1} \right\}}
\newcommand{\Set}[2]{\left\{ #1 \;\middle\vert\; #2 \right\}}

\def\pr{\texttt{Pr}}

\def\total{\texttt{total}}
\def\mypar#1{\medskip\noindent\textbf{#1}}

\newcommand{\tldO}{\scalerel*{\widetilde{O}}{j^2}}%
\newcommand{\tldOmega}{\scalerel*{\widetilde{\Omega}}{j^2}}%
\newcommand{\tldTheta}{\scalerel*{\widetilde{\Theta}}{j^2}}%

\newcommand{\prob}[1]{\textup{\fontencoding{T1}\fontfamily{ppl}\fontseries{m
}\fontshape{n}\selectfont #1}\xspace}

\def\sI{\script{I}}

\def\sC{\script{C}}
\def\sM{\script{M}}

\newcommand{\GSet}{\mathsf{U}}%
\newcommand{\GSetA}{\mathsf{V}}%
\newcommand{\GSample}{\mathsf{S}}%
\newcommand{\GLeftover}{\mathsf{L}}%
\newcommand{\sS}{\mathcal{F}}
\newcommand{\sSA}{\mathcal{H}}
\newcommand{\Cover}{\mathcal{C}}
\newcommand{\CoverA}{\mathcal{D}}
\newcommand{\sSSample}{\mathcal{F}_\GSample}

\def\NP{\ensuremath{\mathrm{\mathbf{NP}}}}

\newcommand{\SarielThanks}[1]{\thanks{Department of Computer Science;
      University of Illinois; 201 N. Goodwin Avenue; Urbana, IL,
      61801, USA; %
      { \url{http://sarielhp.org/}; %
         \href{mailto:sariel@illinois.edu}{sariel@illinois.edu}%
         .} #1}}

\newcommand{\PiotrThanks}[1]{%
   \thanks{CSAIL; EECS;
      MIT; 32 Vassar Street; Cambridge, MA,
      02139, USA; %
       \href{mailto:indyk@mit.edu}{indyk@mit.edu}. #1}}%

\newcommand{\SepidehThanks}[1]{%
   \thanks{CSAIL; EECS;
      MIT; 32 Vassar Street; Cambridge, MA,
      02139, USA; %
       \href{mailto:mahabadi@mit.edu}{mahabadi@mit.edu}.%
      #1}%
}

\newcommand{\AliThanks}[1]{%
   \thanks{CSAIL; EECS;
      MIT; 32 Vassar Street; Cambridge, MA,
      02139, USA; %
      \href{mailto:vakilian@mit.edu}{vakilian@mit.edu}.%
      #1}%
}

\newcommand{\algIterSC}{\Algorithm{iterSetCover}\xspace}
\renewcommand{\Algorithm}[1]{{\AlgorithmI{#1}}}
\newcommand{\ds}{\displaystyle}%
\newcommand{\offlineSC}{\Algorithm{algOfflineSC}\xspace}%
\newcommand{\recoverRanBit}{\Algorithm{algRecoverBit}\xspace}
\newcommand{\existsDisj}{\Algorithm{algExistsDisj}\xspace}
\newcommand{\range}{\mathsf{r}}%
\renewcommand{\Re}{\mathbb{R}}

\newcommand{\INP}[1]{#1}
\newcommand{\InPODS}[1]{}
\newcommand{\si}[1]{#1}%

\begin{document}

\title{Towards Tight Bounds for the Streaming Set Cover Problem%
   \footnote{%
      A preliminary version of this paper is to appear in PODS 2016. %
      Work on this paper by SHP was partially supported by NSF AF
      awards CCF-1421231, and CCF-1217462.  Other authors were
      supported by NSF, Simons Foundation and MADALGO --- Center for
      Massive Data Algorithmics --- a Center of the Danish National
      Research Foundation.%
   }%
}

\author{%
   Sariel Har-Peled%
   \SarielThanks{}%
   \and%
   Piotr Indyk%
   \PiotrThanks{}%
   \and%
   Sepideh Mahabadi%
   \SepidehThanks{}%
   \and%
   Ali Vakilian%
   \AliThanks{}%
}

\date{}

\maketitle

\begin{abstract}
    We consider the classic \probSetCover problem in the data stream
    model.  For $n$ elements and $m$ sets ($m\geq n$) we give a
    $O(1/\delta)$-pass algorithm with a strongly sub-linear
    $\tldO(mn^{\delta})$ space and logarithmic approximation
    factor\footnote{%
       The notations $\tldO(f(n,m))$ and $\tldOmega(f(n,m))$ hide
       polylogarithmic factors. Formally $\tldO(f(n,m))$ and
       $\tldOmega(f(n,m))$ are short form for
       $O\pth{\bigl. f(n,m) \polylog(n,m)}$ and
       $\Omega\pth{ \bigl. f(n,m)/ \polylog(n,m)}$, respectively.}.
    This yields a significant improvement over the earlier algorithm
    of Demaine \etal \cite{dimv-sccsc-14} that uses exponentially
    larger number of passes.  We complement this result by showing
    that the tradeoff between the number of passes and space exhibited
    by our algorithm is tight, at least when the approximation factor
    is equal to $1$.  Specifically, we show that any algorithm that
    computes set cover exactly using $({1 \over 2\delta}-1)$ passes
    must use $\tldOmega(mn^{\delta})$ space in the regime of
    $m=O(n)$. Furthermore, we consider the problem in the geometric
    setting where the elements are points in $\Re^2$ and sets are
    either discs, axis-parallel rectangles, or fat triangles in the
    plane, and show that our algorithm (with a slight modification)
    uses the optimal $\tldO(n)$ space to find a logarithmic
    approximation in $O(1/\delta)$ passes.

     Finally, we show that any randomized one-pass algorithm that
    distinguishes between covers of size 2 and 3 must use a linear
    (i.e., $\Omega(mn)$) amount of space. 
    This is the first result showing that a randomized, approximate 
    algorithm cannot achieve a space bound that is sublinear in the input size. 

	This indicates that using multiple passes might be necessary in order
    to achieve sub-linear space bounds for this problem while
    guaranteeing small approximation factors.
\end{abstract}

\section{Introduction}

The \probSetCover problem is a classic combinatorial optimization
task.  Given a ground set of $n$ elements
$\GSet=\set{e_1, \cdots, e_n}$, and a family of $m$ sets
$\sS = \set{\range_1, \dots, \range_m}$ where $m\geq n$, the goal is to select a subset
$\sI \subseteq \sS$ such that $\sI$ {\em covers} $\GSet$, i.e.,
$\GSet \subseteq \bigcup_{S\in\sI} S$, and the number of the sets in
$\sI$ is as small as possible.  \probSetCover is a well-studied
problem with applications in many areas, including operations
research~\cite{gw-ceaas-97}, information retrieval and data
mining~\cite{sg-mcsma-09}, web host analysis \cite{ckt-mcmr-10}, and
many others.

Although the problem of finding an optimal solution is NP-complete, a natural greedy
algorithm which iteratively picks the ``best'' remaining set is widely
used.  The algorithm often finds solutions that are very close to
optimal. Unfortunately, due to its sequential nature, this algorithm
does not scale very well to massive data sets (e.g., see Cormode
\etal~\cite{ckw-scavl-10} for an experimental evaluation). This
difficulty has motivated a considerable research effort whose goal was
to design algorithms that are capable of handling large data
efficiently on modern architectures. Of particular interest are
\emph{data stream} algorithms, which compute the solution using only a
small number of sequential passes over the data using a limited
memory.  In the \emph{streaming} \probSetCover
problem~\cite{sg-mcsma-09}, the set of elements $\GSet$ is stored in
the memory in advance; the sets $\range_1, \cdots, \range_m$ are stored
consecutively in a read-only repository and an algorithm can access
the sets only by performing sequential scans of the
repository. However, the amount of read-write memory available to the
algorithm is limited, and is smaller than the input size (which could
be as large as $mn$). The objective is to design an efficient
approximation algorithm for the \probSetCover problem that performs
few passes over the data, and uses as little memory as possible.

The last few years have witnessed a rapid development of new streaming
algorithms for the \probSetCover problem, in both theory and applied
communities, see \cite{sg-mcsma-09, ckw-scavl-10, kmvv-fgams-13,
   er-sssc-14, dimv-sccsc-14, cw-igpcs-16}. \figref{f:table} presents
the approximation and space bounds achieved by those algorithms, as
well as the lower bounds\footnote{Note that the simple greedy
   algorithm can be implemented by either storing the whole input (in
   one pass), or by iteratively updating the set of yet-uncovered
   elements (in at most $n$ passes).}.

\begin{figure}[t]
    \begin{center}
        \begin{tabular}{|c|c|c|c|c|c|}
          \hline
          {\bf Result} & {\bf Approx\INP{imation}}& {\bf Passes} & {\bf Space} & {\bf R}\\
          \hline\hline
          Greedy
                       & $\ln n$ & $1$ & $O(m n)$ & \\%
          algorithm%
                       & $\ln n$ & $n$ & $O(n)$ & \\
          \hline
          \cite{sg-mcsma-09} & $O(\log n)$ & $O(\log n)$ & $O(n^2 \ln n)$\!\!\! & \\
          \hline
          \cite{er-sssc-14} & $O(\sqrt{n})$ & 1 & $\tldTheta(n)$ & R\\
          \hline
          \cite{cw-igpcs-16} & $O({n^{\delta}/ \delta})$ & ${1 /
                                                           \delta}-1$
                                                                 & $\tldTheta(n)$ & R\\
          \hline
          \cite{n-ccasp-02} & $\frac{1}{2} \log n$ & $O(\log n)$ & $\tldOmega(m)$ & R\\
          \hline
          \cite{dimv-sccsc-14}& $O(4^{1/\delta}\rho)$ & $O(4^{1/\delta})$ & $\tldO(m n^{\delta})$ & R\\
          \hline
          \cite{dimv-sccsc-14} & $O(1)$ & $O(\log n)$ & $\tldOmega(m n)$ & \\
          \hline \hline
          \thmref{unweight} & $O(\rho/\delta)$ & $2/\delta$ &  $\tldO(m n^{\delta})$ & R\\
          \hline
          \thmref{lower-bound-set-cover} & $3/2$ & 1 & $\Omega(m n)$ & R\\
          
          \hline
          \thmref{reduction}%
                 &%
                   1 & ${1 / 2\delta}-1$%
          &  
                                           $\tldOmega(mn^{\delta})$ & R \\
          \hline \hline
          \small{\prob{Geometric Set Cover}}
          \scriptsize{(\thmref{geometric-set-cover})}
        & $O(\rho_g)$ & $O(1)$ &  $\tldO(n)$ & R\\
          \hline
              \small{\prob{$s$-Sparse Set Cover}}
              \scriptsize{(\thmref{sparse-set-cover})}
          & 1 & ${1 / 2\delta}-1$ &  $\tldOmega(ms)$ & R \\
          \hline%
        \end{tabular}
    \end{center}
    \vspace{-0.5cm}%
    \caption{{Summary of past work and our results.  The
          last column indicates if the scheme is randomized, $\rho$
          denotes the approximation factor of an off-line algorithm
          solving \probSetCover, which is $\ln n$ for the greedy, and
          $1$ for exponential algorithm. Similarly, $\rho_g$ denotes
          the approximation factor of an off-line algorithm solving
          geometric \prob{Set
             Cover}. 
          Finally,
          in the \prob{$s$-Sparse Set Cover} problem, $s \leq n^{\delta}$ denotes 
          an upper bound on the sizes of the input sets. Our lower bounds for 
          \probSetCover and \prob{$s$-Sparse Set Cover} hold for $m=O(n)$.
           Moreover,~\cite{er-sssc-14}
          and~\cite{cw-igpcs-16} proved that their algorithms are
          tight. 
          Here, and in the rest of the paper, all $\log$ are in
          base two.}}
    \figlab{f:table}
    \InPODS{\vspace{-0.7cm}}%
\end{figure}

\mypar{Related work.} %
The semi-streaming \probSetCover problem was first studied by Saha and
Getoor \cite{sg-mcsma-09}. Their result for \prob{Max $k$-Cover}
problem implies a $O(\log n)$-pass $O(\log n)$-approximation algorithm
for the \probSetCover problem that uses $\tldO(n^2)$ space.  
Adopting the standard greedy algorithm of \prob{Set Cover} with a thresholding
technique leads to $O(\log n)$-pass $O(\log n)$-approximation using $\tldO(n)$ space.
In $\tldO(n)$ space regime, Emek and Rosen studied designing one-pass
streaming algorithms for the \probSetCover problem \cite{er-sssc-14}
and gave a \emph{deterministic} greedy based
$O(\sqrt{n})$-approximation for the problem. Moreover they proved that
their algorithm is tight, even for randomized algorithms. The lower/upper bound results
of \cite{er-sssc-14} applied also to a generalization of the \probSetCover 
problem, the \prob{$\eps$-Partial Set Cover($\GSet, \sS$)}
problem in which the goal is to cover $(1-\eps)$ fraction of elements
$\GSet$ and the size of the solution is compared to the size of an optimal 
cover of \prob{Set Cover($\GSet, \sS$)}. 
Very recently, Chakrabarti and Wirth extended the result of
\cite{er-sssc-14} and gave a trade-off streaming algorithm for the
\probSetCover problem in \emph{multiple passes} \cite{cw-igpcs-16}.
They gave a \emph{deterministic} algorithm with $p$ passes over the
data stream that returns a $(p+1)n^{1/(p+1)}$-approximate solution of
the \probSetCover problem in $\tldO(n)$ space.  Moreover they proved
that achieving $0.99 n^{1/(p+1)}/(p+1)^2$ in $p$ passes using
$\tldO(n)$ space is not possible even for randomized protocols which
shows that their algorithm is tight up to a factor of $(p+1)^3$. Their
result also works for the \prob{$\eps$-Partial Set Cover} problem.

In a different regime which was first studied by Demaine \etal, the
goal is to design a ``low" approximation algorithms (depending on 
the computational model, it could be $O(\log n)$ or $O(1)$) in the smallest 
possible space \cite{dimv-sccsc-14}. They proved that any constant 
pass \emph{deterministic} $(\log n/2)$-approximation algorithm for 
the \probSetCover problem requires $\tldOmega(mn)$ space. It shows
that unlike the results in $\tldO(n)$-space regime, to obtain a
sublinear ``low" approximation streaming algorithm for the \probSetCover
problem in a constant number of passes, using \emph{randomness} 
is necessary. Moreover, \cite{dimv-sccsc-14} presented a 
$O(4^{1/\delta})$-approximation algorithm that makes $O(4^{1/\delta})$
passes and uses $\tldO(mn^{\delta})$ memory space.

The \prob{Set Cover} problem is not polynomially solvable even in the restricted instances 
with points in $\Re^2$ as elements, and geometric objects (either all disks or 
axis parallel rectangles or fat triangles) in plane as sets 
\cite{fg-oafac-88, fpt-opcpn-81,
   hq-aapel-15}.
As a result, there has been a large body of work on designing approximation algorithms 
for the geometric \prob{Set Cover} problems. See for example 
\cite{mrr-sahsgsc-14,ap-nlagh-14,
   aes-ssena-10, cv-iaagsc-07} and references therein.

\subsection{Our results}
Despite the progress outlined above, however, some basic questions
still remained open. In particular:
\INP{\smallskip}%
\begin{compactenum}[\INP{\quad}(A)]
    \item Is it possible to design a {\em single} pass streaming
    algorithm with a ``low'' approximation factor\footnote{Note that
       the lower bound in~\cite{dimv-sccsc-14} excluded this
       possibility only for deterministic algorithms, while the upper
       bound in~\cite{er-sssc-14, cw-igpcs-16} suffered from a
       polynomial approximation factor. } that uses sublinear (i.e., $o(mn)$) space?
    \item If such single pass algorithms are not possible, what are
    the achievable trade-offs between the number of passes and space
    usage?
    \item Are there special instances of the problem for which more efficient algorithms can be designed?
\end{compactenum}
\INP{\smallskip}%
In this paper, we make a significant progress on each of these
questions. Our upper and lower bounds are depicted in
\figref{f:table}.

On the algorithmic side, we give a $O(1/\delta)$-pass algorithm with a
strongly sub-linear $\tldO(mn^{\delta})$
space
    and logarithmic
approximation factor.  This yields a significant improvement over the
earlier algorithm of Demaine \etal~\cite{dimv-sccsc-14} which used
exponentially larger number of passes.  The trade-off offered by our
algorithm matches the lower bound of Nisan~\cite{n-ccasp-02} that
holds at the endpoint of the trade-off curve, i.e., for
$\delta=\Theta(1/\log n)$, up to poly-logarithmic factors in
space\footnote{Note that to achieve a logarithmic approximation ratio
   we can use an off-line algorithm with the approximation ratio
   $\rho=1$, i.e., one that runs in exponential time (see
   \thmref{algorithm-tightness}).}. Furthermore, our algorithm is
very simple and succinct, and therefore easy to implement and deploy.

Our algorithm exhibits a natural tradeoff between the number of passes
and space, which resembles tradeoffs achieved for other
problems~\cite{gm-lbqer-07,gm-tlbsp-08,go-slbmg-13}.  It is thus
natural to conjecture that this tradeoff might be tight, at least for
``low enough'' approximation factors.  We present the first step in
this direction by showing a lower bound for the case when the
approximation factor is equal to $1$, i.e., the goal is to compute the
optimal set cover.  In particular, by an information theoretic lower
bound, we show that any \emph{streaming} algorithm that computes set
cover using $({1 \over 2\delta}-1)$ passes must use
$\tldOmega(mn^{\delta})$ space (even assuming exponential
computational power) in the regime of $m=O(n)$. Furthermore, we show that a stronger lower bound 
holds if all the input sets are sparse, that is if their cardinality is at most
$s$. We prove a lower bound of $\tldOmega(ms)$ for
$s=O(n^{\delta})$ and $m=O(n)$.

We also consider the problem in the geometric setting in which the
elements are points in $\Re^2$ and sets are either discs,
axis-parallel rectangles, or fat triangles in the plane. We show that
a slightly modified version of our algorithm achieves the \emph{optimal}
$\tldO(n)$ space to find an $O(\rho)$-approximation in
$O(1)$ passes.

Finally, we show that any randomized one-pass algorithm that
distinguishes between covers of size 2 and 3 must use a linear (i.e.,
$\Omega(mn)$) amount of space. 
 This is the first result showing that a randomized, approximate algorithm cannot achieve a sub-linear space bound. 

Recently Assadi \etal~\cite{aky-tbspscscp-16} generalized this lower bound
to any approximation ratio $\alpha=O(\sqrt{n})$. More precisely
they showed that approximating \prob{Set Cover} within any factor
$\alpha = O(\sqrt{n})$ in a single pass requires 
$\Omega({mn \over \alpha})$ space.
 


\noindent%
\textbf{Our techniques: Basic idea.}  %
Our algorithm is based on the idea that whenever a large enough set is
encountered, we can immediately add it to the cover. Specifically, we
guess (up to factor two) the size of the optimal cover $k$. Thus, a
set is ``large'' if it covers at least $1/k$ fraction of the remaining
elements. A small set, on the other hand, can cover only a ``few''
elements, and we can store (approximately) what elements it covers by
storing (in memory) an appropriate random sample. At the end of the
pass, we have (in memory) the projections of ``small'' sets onto the
random sample, and we compute the optimal set cover for this projected
instance using an offline solver. By carefully choosing the size of
the random sample, this guarantees that only a small fraction of the
set system remains uncovered. The algorithm then makes an additional
pass to find the residual set system (i.e., the yet uncovered
elements), making two passes in each iteration, and continuing to the
next iteration.

Thus, one can think about the algorithm as being based on a simple iterative
``dimensionality reduction'' approach. Specifically, in two passes
over the data, the algorithm selects a ``small'' number of sets that
cover all but $n^{-\delta}$ fraction of the uncovered elements, while
using only $\tldO(m n^{\delta})$ space.  By performing the reduction
step $1/\delta$ times we obtain a complete cover.  The dimensionality
reduction step is implemented by computing a small cover for a
\emph{random subset} of the elements, which also covers the vast
majority of the elements in the ground set.  This ensures that the
remaining sets, when restricted to the random subset of the elements,
occupy only $\tldO(m n^{\delta})$ space.  As a result the procedure
avoids a complex set of recursive calls as presented in Demaine
\etal~\cite{dimv-sccsc-14}, which leads to a simpler and more
efficient algorithm.

\noindent%
\textbf{Geometric results.} %
Further using techniques and results from computational geometry we
show how to modify our algorithm so that it achieves \emph{almost} optimal bounds for 
the \probSetCover problem on geometric instances.  In
particular, we show that it gives a $O(1)$-pass $O(\rho)$-approximation algorithm 
using $\tldO(n)$ space when
the elements are points in $\Re^2$ and the sets are either discs, axis
parallel rectangles, or fat triangles in the plane.  In particular, we
use the following surprising property of the set systems that arise out of
points and disks: the the number of sets is nearly linear as long
as one considers only sets that contain ``a few'' points.

More surprisingly, this property extends, with a twist, to certain
geometric range spaces that might have quadratic number of shallow
ranges. Indeed, it is easy to show an example of $n$ points in the
plane, where there are $\Omega(n^2)$ distinct rectangles, each one
containing exactly two points, see \figref{many:rectangles}. However,
one can ``split'' such ranges into a small number of canonical sets,
such that the number of shallow sets in the canonical set system is
near linear. This enables us to store the small canonical sets encountered 
during the scan explicitly in memory, and still use only near linear space.

\noindent
\begin{figure}[htb]
    \centering
    \includegraphics{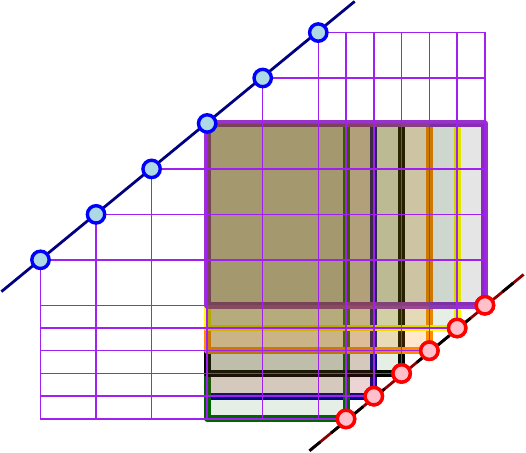}
    \captionof{figure}{Consider two parallel lines in the plane with positive
       slope.  Place $n/2$ points on each line such that all the
       points on the top line lie above and to the left of all the
       points on the bottom line.  Let the set of rectangles for this
       instance be all the rectangles which have a point on the top
       line as their upper left corner and a point on the bottom line
       as their lower right corner. Clearly, we have $n^2/4$ distinct
       rectangles (i.e., sets), each containing two points. As such,
       we cannot afford to store explicitly in memory the set system,
       since it requires too much space. }
    \figlab{many:rectangles}
\end{figure}

We note that the idea of splitting ranges into small canonical ranges is an old
idea in orthogonal range searching. It was used by Aronov
\etal~\cite{aes-ssena-10} for computing small $\eps$-nets for these
range spaces. The idea in the form we use, was further formalized by
Ene \etal~\cite{ehr-gpnuc-12}.

\noindent%
\textbf{Lower bounds.} %
The lower bounds for multi-pass algorithms for the \probSetCover problem 
are obtained via a careful reduction from \prob{Intersection Set Chasing}. The
latter problem is a communication complexity problem where $n$ players
need to solve a certain ``set-disjointness-like" problem in $p$
rounds. A recent paper \cite{go-slbmg-13} showed that this problem
requires $n^{1+\Omega(1/p)}\over p^{O(1)}$ bits of communication
complexity for $p$ rounds. This yields our desired trade-off 
of $\tldOmega(mn^{\delta})$ space in $1/2\delta$ passes for 
exact protocols for \prob{Set Cover} in 
the communication model and hence in the streaming model for $m=O(n)$.
Furthermore, we show a stronger lower
bound on memory space of sparse instances of \probSetCover in which
all input sets have cardinality at most $s$. By a reduction from
a variant of \prob{Equal Pointer Chasing} which maps the problem to a 
\emph{sparse} instance of \prob{Set Cover}, we show that in order to have an 
exact streaming algorithm for \prob{$s$-Sparse Set Cover} with $o(ms)$ space, 
$\Omega(\log n)$ passes is necessary. More precisely, 
any $({1\over 2\delta}-1)$-pass \emph{exact} randomized algorithm for 
\prob{$s$-Sparse Set Cover} requires $\tldOmega(ms)$ memory space, if
$s\leq n^{\delta}$ and $m=O(n)$.

Our single pass lower bound proceeds by showing a lower bound for a
one-way communication complexity problem in which one party (Alice)
has a collection of sets, and the other party (Bob) needs to determine
whether the complement of his set is covered by one of the Alice's
sets. We show that if Alice's sets are chosen at random, then Bob can
decode Alice's input by employing a small collection of ``query''
sets. This implies that the amount of communication needed to solve
the problem is linear in the description size of Alice's sets, which
is $\Omega(mn)$. %
\begin{figure}
    \begin{center}
        \begin{algorithmEnvX}
            \underline{\textbf{\algIterSC}%
               $\pth{ \pth{\GSet, \sS}, \delta}$}: $\Bigl.$ \qquad
            \+\\ %
            \INP{%
               \CodeComment{// Try in parallel all possible
                  ($2$-approx) sizes of optimal cover}\\%
            }%
            \InPODS{ \CodeComment{// $n =\card{\GSet}$}\\}%
            \For $k \in \set{2^{i}\sep 0\leq i \leq \log n}$ \Do
            \textbf{in parallel}: \INP{\qquad%
               \CodeComment{// $n =\card{\GSet}$}}%
            \\
            \> $\sol \leftarrow \emptyset$ \\%
            \> {\bf Repeat} for $1/\delta$ times \\%
            \>\> \INP{\bf Let} $\GSample$ be a sample of $\GSet$ of
            size $c\rho kn^{\delta} \log m\log n$ %
            \\%
            \>\> $\GLeftover\leftarrow \GSample$, \quad
            $\sSSample \leftarrow \emptyset$ \\%
            \>\> \For $\range \in \sS$ \Do \quad%
            \INP{\qquad\qquad\quad} \CodeComment{// By doing one
               pass} \\%
            \>\>\> \If
            $\card{\GLeftover \cap \range} \geq \card{\GSample}/k$
            \Then \quad\CodeComment{// Size Test} \\%
            \>\>\>\> $\sol \leftarrow \sol \cup \brc{\range}$%
            \InPODS{, }%
            \INP{\\%
               \>\>\>\>}
            $\GLeftover \leftarrow \GLeftover\setminus \range$ \\%
            \>\>\>\Else \\%
            \>\>\>\>
            $\sSSample\leftarrow \sSSample\cup \brc{ \range \cap
               \GLeftover}$ %
            \CodeComment{// Store the set $\range \cap \GLeftover$ explicitly in memory} \\[0.1cm]%
            \>\>
            $\CoverA \leftarrow \offlineSC{}(\GLeftover,\sSSample,k)$,
            \qquad $\sol \leftarrow \sol \bigcup\CoverA$ \\%
            \>\>
            $\GSet \leftarrow \GSet \setminus \bigcup_{\range \in
               \sol}\range$ %
            \INP{ \quad}%
            \INP{%
               \InPODS{\\ \>\>}%
               \CodeComment{// By doing additional pass over
                  data}\\[0.1cm]%
            }%
            \InPODS{\\}%
            \Return best $\sol$ computed in all parallel executions.
        \end{algorithmEnvX}%
    \end{center}
    \vspace{-0.5 cm}
    \caption{{A tight streaming algorithm for the (unweighted)
       \probSetCover problem.  Here, \offlineSC{} is an offline solver
       for \probSetCover that provides $\rho$-approximation, and $c$
       is some appropriate constant.}}
 \figlab{unweighted-alg} 
\end{figure}%


\section{Streaming Algorithm\INP{ for Set Cover}}
\seclab{unweighted-set-cover}%

\subsection{Algorithm}

In this section, we design an efficient streaming algorithm for the
\probSetCover problem that matches the lower bound results we already
know about the problem.  In the \probSetCover problem, for a given set
system $(\GSet,\sS)$, the goal is to find a subset
$\sI \subseteq \sS$, such that $\sI$ covers $\GSet$ and its
cardinality is minimum.  In the following, we sketch the \algIterSC
algorithm (see also \figref{unweighted-alg}).

In the \algIterSC algorithm, we have access to the \offlineSC{}
subroutine that solves the given \probSetCover instance offline (using
linear space) and returns a $\rho$-approximate solution where $\rho$
could be anywhere between $1$ and $\Theta(\log n)$ depending on the
computational model one assumes. Under exponential computational
power, we can achieve the optimal cover of the given instance of the
\probSetCover ($\rho=1$); however, under ${\bf P}\neq\NP$ assumption,
$\rho$ cannot be better than $c\cdot \ln n$ where $c$ is a
constant~\cite{f-tasc-98,rs-sbepl-97,ams-acskr-06,m-pgcnsc-12,ds-aapr-14} given polynomial computational
power.

Let $n= \card{\GSet}$ be the initial number of elements in the given
ground set. The \algIterSC algorithm, needs to guess (up to a factor
of two) the size of the optimal cover of $(\GSet,\sS)$.  To this end,
the algorithm tries, in parallel, all values $k$ in
$\set{2^{i}\sep 0\leq i\leq \log n}$. This step will only increase the
memory space requirement by a factor of $\log n$.

Consider the run of the \algIterSC algorithm, in which the guess $k$
is correct (i.e., $\card{\opt} \leq k < 2\card{\opt}$, where $\opt$ is
an optimal solution). The idea is to go through $O(1/\delta)$
iterations such that each iteration only makes two passes and at the
end of each iteration the number of uncovered elements reduces by a
factor of $n^{\delta}$. Moreover, the algorithm is allowed to use
$\tldO(mn^{\delta})$ space.

In each iteration, the algorithm starts with the current ground set of
uncovered elements $\GSet$, and copies it to a \emph{leftover} set
$\GLeftover$. Let $\GSample$ be a large enough uniform sample of
elements $\GSet$.  In a single pass, using $\GSample$, we estimate the
size of all large sets in $\sS$ and add $\range \in \sS$ to the
solution $\sol$ immediately (thus avoiding the need to store it in
memory). Formally, if $\range$ covers at least
$\Omega(\card{\GSet}/k)$ yet-uncovered elements of $\GLeftover$ then
it is a heavy set, and the algorithm immediately adds it to the output
cover. Otherwise, if a set is small, i.e., its covers less than
$\card{\GSet}/k$ uncovered elements of $\GLeftover$, the
algorithm stores the set $\range$ in memory. Fortunately, it is enough
to store its projection over the sampled elements explicitly (i.e.,
$\range \cap \GLeftover$) -- this requires remembering only the
$O(\card{\GSample}/k)$ indices of the elements of
$\range \cap \GLeftover$.

In order to show that a solution of the \probSetCover problem over
the sampled elements is a good cover of the initial \probSetCover
instance, we apply the \emph{relative ($p,\eps$)-approximation}
sampling result of \cite{hs-rag-11} (see
\defref{def:relative-approximation}) and it is enough for $\GSample$
to be of size $\tldO(\rho k n^{\delta})$.
Using relative $(p,\eps)$-approximation sampling, we show that
after two passes the number of uncovered elements is reduced by a factor
of $n^{\delta}$. Note that the relative $(p,\eps)$-approximation sampling improves 
over the \emph{Element Sampling} technique used in \cite{dimv-sccsc-14} with respect to the number of passes. 

Since in each iteration we pick $O(\rho k)$ sets and the number of
uncovered elements decreases by a factor of $n^{\delta}$, after
$1/\delta$ iterations the algorithm picks $O(\rho k/\delta)$ sets and
covers all elements. Moreover, the memory space of the whole algorithm
is $\tldO(\rho mn^{\delta})$ (see
\lemref{unweight-space}).

\subsection{Analysis}

In the rest of this section we prove that the
\algIterSC algorithm with high probability returns a
$O(\rho/\delta)$-approximate solution of \prob{Set Cover($\GSet,\sS$)}
in $2/\delta$ passes using $\tldO(mn^{\delta})$ memory space.
\begin{lemma}
    \lemlab{unweight-pass}%
    The number of passes the \algIterSC algorithm
    makes is $2/\delta$.
\end{lemma}
\begin{proof}
    In each of the $1/\delta$ iterations of the
    \algIterSC algorithm, the algorithm makes two
    passes. In the first pass, based on the set of sampled elements
    $\GSample$, it decides whether to pick a set or keep its
    projection over $\GSample$ (i.e., $\range \cap \GLeftover$) in the memory. 
    Then the algorithm calls \offlineSC{} which does not require any passes
    over $\sS$. The second pass is for computing the set of uncovered
    elements at the end of the iteration.  We need this pass because
    we only know the projection of the sets we picked in the current
    iteration over $\GSample$ and not over the original set of
    uncovered elements. Thus, in total we make $2/\delta$ passes. Also
    note that for different guesses for the value of $k$, we run the
    algorithm in parallel and hence the total number of passes remains
    $2/\delta$.
\end{proof}

\begin{lemma}%
    \lemlab{unweight-space}%
    The memory space used by the \algIterSC algorithm
    is $\tldO(mn^{\delta})$.
\end{lemma}
\begin{proof}
    In each iteration of the algorithm, it picks during the first pass at
    most $m$ sets (more precisely at most $k$ sets) which requires $O(m\log m)$
    memory.  Moreover, in the first pass we keep the projection of the
    sets whose projection over the uncovered sampled elements has size
    at most ${\card{\GSample} / k}$. Since there are at most $m$ such
    sets, the total required space for storing the projections is
    bounded by
    \begin{math}
        O\pth{\rho mn^{\delta} \log m \log n}.
    \end{math}

    Since in the second pass the algorithm only updates the set of
    uncovered elements, the amount of space required in the second
    pass is $O(n)$. Thus, the total required space to perform each
    iteration of the \algIterSC algorithm is
    $\tldO(mn^{\delta})$. Moreover, note that the algorithm does not
    need to keep the memory space used by the earlier iterations;
    thus, the total space consumed by the algorithm is
    $\tldO(mn^{\delta})$.
\end{proof}
Next we show the sets we picked before calling \offlineSC{} has large
size on $\GSet$.
\begin{lemma} %
    \lemlab{large-set}%
    With probability at least $1-m^{-c}$ all sets that pass the ``Size
    Test" in the \algIterSC algorithm have size at least
    $\card{\GSet} / ck$.
\end{lemma}
\begin{proof}
    Let $\range$ be a set of size less than $\card{\GSet} / ck$. In
    expectation, $\card{\range \cap \GSample}$ is less than
    \begin{math}
        \pth{\card{\GSet} / ck} \cdot {\pth{\card{\GSample} / \card{\GSet}}}%
        =%
        {\rho n^{\delta} \log m \log n}.
    \end{math}
    By Chernoff bound for large enough $c$,
    \begin{align*}
        \pr(\card{\range \cap \GSample} \geq c{\rho n^{\delta} \log m \log n})%
        \leq%
        m^{-(c+1)}.
    \end{align*}
    Applying the union bound, with probability at least $1-m^{-c}$,
    all sets passing ``Size Test'' have size at least
    ${\card{\GSet}/(ck)}$.
\end{proof}

In what follows we define the \emph{relative
   $(p,\eps)$-approximation} sample of a set system and mention
the result of Har-Peled and Sharir \cite{hs-rag-11} on the minimum
required number of sampled elements to get a relative
$(p,\eps)$-approximation of the given set system.
\begin{definition} %
    \deflab{def:relative-approximation}%
    Let $(\GSetA,\sSA)$ be a set system, i.e., $\GSetA$ is a set of
    elements and $\sSA \subseteq 2^\GSetA$ is a family of subsets of the
    ground set $\GSetA$. For given parameters $0< \eps , p < 1$, a
    subset $Z \subseteq \GSetA$ is a \emph{relative
       $(p,\eps)$-approximation} for $(\GSetA,\sSA)$, if for each
    $\range\in \sSA$, we have that if $\card{\range} \geq p \card{\GSetA}$ then
    \begin{align*}
        \ds%
        (1-\eps) {\card{\range}\over \card{\GSetA}} \leq {\card{\range\cap Z}
           \over \card{Z}} \leq (1+\eps){\card{\range} \over\card{\GSetA}}.%
    \end{align*} 
    If the range is light (i.e., $\card{\range} < p \card{\GSetA}$) then it
    is required that
    \begin{align*}
        \ds%
        {\card{\range}\over \card{\GSetA}}-\eps p%
        \leq%
        {\card{\range\cap Z} \over \card{Z}}%
        \leq%
        {\card{\range} \over\card{\GSetA}}+\eps p.%
    \end{align*}
    Namely, $Z$ is $(1\pm\eps)$-multiplicative good estimator for the
    size of ranges that are at least $p$-fraction of the ground set.
\end{definition}

The following lemma is a simplified variant of a result in Har-Peled
and Sharir \cite{hs-rag-11} -- indeed, a set system with $M$ sets, can
have VC dimension at most $\log M$. This simplified form also follows
by a somewhat careful but straightforward application of Chernoff's
inequality.
\begin{lemma}%
    \lemlab{relative-approx}%
    Let $(\GSet,\sS)$ be a finite set system, and $p,\eps, q$ be
    parameters.  Then, a random sample of $\GSet$ such that 
    \begin{math}
        \card{\GSet} = {c'\over \eps^2 p}\pth{\log \card{\sS} \log {1\over p} +
           \log{1\over q}},
    \end{math}
    for an absolute constant $c'$ is a relative
    $(p,\eps)$-approximation, for all ranges in $\sS$, with
    probability at least $(1-q)$.
\end{lemma}

\begin{lemma}
    \lemlab{completeness} %
    Assuming $\card{\opt}\leq k \leq 2\card{\opt}$, after any
    iteration, with probability at least $1-m^{1-c/4}$ the number of
    uncovered elements decreases by a factor of $n^{\delta}$, and this
    iteration adds $O(\rho\card{\opt})$ sets to the output cover.
\end{lemma}

\newcommand{\FF}{\mathcal{G}}%

\begin{proof}
    Let $\GSetA \subseteq \GSet$ be the set of uncovered elements at the
    beginning of the iteration and note that the total number of sets
    that is picked during the iteration is at most $(1+\rho)k$ (see~\lemref{large-set}).
    Consider all possible such covers, that is
    $\FF = \{ \sS' \subseteq \sS|$ $\bigl.\card{\sS'}\leq
       (1+\rho)k\}$,
    and observe that $\card{\FF} \leq m^{(1+\rho)k}$.  Let $\sSA$ be the
    collection that contains all possible sets of uncovered elements
    at the end of the iteration, defined as
    \begin{math}
        \sSA=\Set{\GSetA \setminus \bigcup_{\range \in \Cover}
           \range}{ \Cover \in \FF}.
    \end{math}
    Moreover, set $p = {2/ n^{\delta}}$, $\eps = 1/2$ and $q=m^{-c}$
    and note that $\card{\sSA} \leq \card{\FF} \leq m^{(1+\rho)k}$.
    Since
    \begin{math}
        {c'\over \eps^2 p}(\log \card{ \sSA }%
        \log {1\over p} + \log{1\over q})%
        \leq%
        c\rho kn^\delta\log m \log n%
        =%
        \card{\GSample}
    \end{math}    
    for large enough $c$, by \lemref{relative-approx}, $\GSample$ is a
    relative $(p,\eps)$-approximation of $(\GSetA,\sSA)$ with $(1-q)$ probability.  Let
    $\CoverA \subseteq \sS$ be the collection of sets picked during
    the iteration which covers all elements in $\GSample$.  Since
    $\GSample$ is a relative $(p,\eps)$-approximation sample of
    $(\GSetA,\sSA)$ with probability at least $1-m^{-c}$, the number
    of uncovered elements of $\GSetA$ (or $\GSet$) by $\CoverA$ is at
    most $\eps p\card{\GSetA} = {\card{\GSet}/n^{\delta}}$.

    Hence, in each iteration we pick $O(\rho k)$ sets and at the end
    of iteration the number of uncovered elements reduces by
    $n^{\delta}$.
\end{proof}

\begin{lemma}
    \lemlab{unweight-approx} %
    The \algIterSC algorithm computes a set cover of $(\GSet,\sS)$,
    whose size is within a $O(\rho/\delta)$ factor of the size of an
    optimal cover with probability at least $1-m^{1-c/4}$.
\end{lemma}
\begin{proof}
    Consider the run of \algIterSC for which the value of $k$ is
    between $\card{\opt}$ and $2\card{\opt}$.  In each of the
    $(1/\delta)$ iterations made by the algorithm, by
    \lemref{completeness}, the number of uncovered elements decreases
    by a factor of $n^{\delta}$ where $n$ is the number of initial
    elements to be covered by the sets.  Moreover, the number of sets
    picked in each iteration is $O(\rho k)$.  Thus after $(1/\delta)$
    iterations, all elements would be covered and the total number of
    sets in the solution is $O(\rho\card{\opt} /\delta)$. Moreover by
    \lemref{completeness}, the success probability of all the
    iterations, is at least
    $1-{1\over \delta m^{c/4}}\geq 1-(1/m)^{{c\over 4}-1}$.
\end{proof}

\begin{theorem} %
    \thmlab{unweight}%
    \thmlab{algorithm-tightness}%
    The $\algIterSC(\GSet, \sS, \delta)$ algorithm
    makes $2/ \delta$ passes, uses $\tldO(mn^{\delta})$ memory space,
    and finds a $O({\rho/\delta})$-approximate solution of the
    \probSetCover problem with high probability.

    Furthermore, given enough number of passes the \algIterSC algorithm matches the known lower
    bound on the memory space of the streaming \probSetCover
    problem up to a $\polylog(m)$ factor where $m$ is the number of
    sets in the input.
\end{theorem}
\begin{proof}
    The first part of the proof implied by \lemref{unweight-pass},
    \lemref{unweight-space}, and \lemref{unweight-approx}.

    As for the lower bound, note that by a result of Nisan
    \cite{n-ccasp-02}, any randomized
    ($\log n \over 2$)-\si{approxima}\-\si{tion} protocol for
    \prob{Set Cover($\GSet,\sS$)} in the one-way communication model
    requires $\Omega(m)$ bits of communication, no matter how many
    number of rounds it makes. This implies that any randomized
    $O(\log n)$-pass, ($\log n \over 2$)-approximation algorithm for
    \prob{Set Cover($\GSet,\sS$)} requires $\tldOmega(m)$ space, even
    under the exponential computational power assumption.

    By the above, the \algIterSC algorithm makes $O(1/\delta)$ passes
    and uses $\tldO(mn^{\delta})$ space to return a
    $O({1\over\delta})$-approximate solution under the exponential
    computational power assumption ($\rho=1$). Thus by letting
    $\delta=c/\log n$, we will have a
    $({\log n\over 2})$-approximation streaming algorithm using
    $\tldO(m)$ space which is optimal up to a factor of $\polylog(m)$.
\end{proof}

\thmref{algorithm-tightness} provides a strong indication that our
trade-off algorithm is optimal.


\section{Lower Bound for Single Pass\INP{ Algorithms}}
\seclab{lower-bound}%

In this section, we study the \probSetCover problem in the
two-party communication model and give a tight lower bound on the
communication complexity of the randomized protocols solving the
problem in a single round.  In the two-party \probSetCover, we are
given a set of elements $\GSet$ and there are two players Alice and Bob
where each of them has a collection of subsets of $\GSet$, $\sS_{A}$ and
$\sS_{B}$. The goal for them is to find a minimum size cover
$\sC \subseteq \sS_{A} \cup \sS_{B}$ covering $\GSet$ while communicating
the fewest number of bits from Alice to Bob (In this model Alice communicates 
to Bob and then Bob should report a solution).

Our main lower bound result for the single pass protocols for \prob{Set Cover} is the following theorem which implies
that the naive approach in which one party sends all of its sets to
the the other one is optimal.
\begin{theorem}%
    \thmlab{main}%
    Any single round randomized protocol that approximates \prob{Set
       Cover$(\GSet,\sS)$} within a factor better than $3/2$ and error
    probability $O(m^{-c})$ requires $\Omega(mn)$ bits of
    communication where $n = \card{\GSet}$ and $m=\card{\sS}$ and $c$ is
    a sufficiently large constant.
\end{theorem}

We consider the case in which the parties want to decide whether there
exists a cover of size $2$ for $\GSet$ in $\sS_{A} \cup \sS_{B}$ or not. If any
of the parties has a cover of size at most $2$ for $\GSet$, then it
becomes
trivial. 
Thus the question is whether there exist $\range_a \in \sS_{A}$ and
$\range_b \in \sS_{B}$ such that $\GSet \subseteq \range_a \cup \range_b$.

A key observation is that to decide whether there exist
$\range_{a} \in \sS_{A}$ and $\range_{b} \in \sS_{B}$ such that
$\GSet \subseteq \range_{a} \cup \range_{a}$, one can instead check whether there
exists $\range_{a} \in \sS_{A}$ and $\range_{b} \in \sS_{B}$ such that
$\bar{\range_{a}} \cap \bar{\range_{b}} = \emptyset$.  In other words we need to
solve \prob{OR} of a series of two-party \prob{Set Disjointness}
problems.  In two-party \prob{Set Disjointness} problem, Alice and Bob are given
subsets of $\GSet$, $\range_{a}$ and $\range_{b}$ and the goal is to decide whether
$\range_{a} \cap \range_{b}$ is empty or not with the fewest possible bits of
communication.  \prob{Set Disjointness} is a well-studied problem in
the communication complexity and it has been shown that any randomized
protocol for \prob{Set Disjointness} with $O(1)$ error probability
requires $\Omega(n)$ bits of communication where $n=\card{\GSet}$
\cite{bjks-isads-04,ks-pccsi-92,r-dcd-92}.

We can think of the following extensions of the \prob{Set
   Disjointness} problem.

\noindent%
\begin{compactenum}
    \item[{\it Many vs One:\INP{\\}}] In this variant, Alice has $m$
    subsets of $\GSet$, $\sS_A$ and Bob is given a single set
    $\range_{b}$. The goal is to determine whether there exists a set
    $\range_{a} \in \sS_A$ such that
    $\range_{a} \cap \range_{b} = \emptyset$.

    \item[{\it Many vs Many:\INP{\\}}] In this variant, each of Alice
    and Bob are given a collection of subsets of $\GSet$ and the goal
    for them is to determine whether there exist $\range_{a}\in \sS_A$
    and $\range_{b}\in \sS_B$ such that
    $\range_{a} \cap \range_{b} = \emptyset$.
\end{compactenum}

\noindent%
Note that deciding whether two-party \probSetCover has a cover of
size $2$ is equivalent to solving the \prob{(Many vs Many)-Set
   Disjointness} problem. Moreover, any lower bound for \prob{(Many vs
   One)-Set Disjointness} clearly implies the same lower bound for the
\prob{(Many vs Many)-Set Disjointness} problem.  In the following
theorem we show that any single-round randomized protocol that solves
\prob{(Many vs One)-Set Disjointness($m,n$)} with $O(m^{-c})$ error
probability requires $\Omega(mn)$ bits of communication.
\begin{theorem}%
    \thmlab{rand-all-one-set-disjointness}%

    Any randomized protocol for \prob{(Many vs One)-Set
       Disjointness$(m,n)$} with error probability that is
    $O( m^{-c})$ requires $\Omega(mn)$ bits of communication if
    $n\geq c_1\log m$ where $c$ and $c_1$ are large enough constants.
\end{theorem}
The idea is to show that if there exists a single-round randomized
protocol for the problem with $o(mn)$ bits of communication and error
probability $O(m^{-c})$, then with constant probability one can
distinguish $\Omega(2^{mn})$ distinct inputs using $o(mn)$ bits which is
a contradiction. 

Suppose that Alice has a collection of $m$ uniformly and independently
random subsets of $\GSet$ (in each of her subsets the probability that
$e\in\GSet$ is in the subset is $1/2$).  Lets assume that there exists a
{\it single round} protocol $\mathbf{I}$ for \prob{(Many vs One)-Set
   Disjointness($n, m$)} with error probability $O(m^{-c})$ using
$o(mn)$ bits of communication.  Let \existsDisj be Bob's
algorithm in protocol
$\mathbf{I}$. 
Then we show that one can recover $mn$ random bits with constant
probability using \existsDisj subroutine and the message
$s$ sent by the first party in protocol $\mathbf{I}$. The
\recoverRanBit which is shown in
\figref{recover-random-set-alg}, is the algorithm to recover random
bits using protocol $\mathbf{I}$ and \existsDisj.

To this end, Bob gets the message $s$ communicated by protocol
$\mathbf{I}$ from Alice and considers all subsets of size $c_1\log m$
and $c_1\log m +1$ of $\GSet$. Note that $s$ is communicated only once
and thus the same $s$ is used for all queries that Bob makes.  Then at
each step Bob picks a random subset $\range_{b}$ of size $c_1\log m$ of $\GSet$
and solve the \prob{(Many vs One)-Set Disjointness} problem with input
$(\sS_A,\range_{b})$ by running $\existsDisj(s,\range_{b})$.  Next we
show that if $\range_{b}$ is disjoint from a set in $\sS_A$, then with high
probability there is \emph{exactly} one set in $\sS_A$ which is
disjoint from $\range_{b}$ (see \lemref{uniqueness}). Thus once Bob finds
out that his query, $\range_{b}$, is disjoint from a set in $\sS_A$, he can
query all sets $\range_{b}^{+} \in \set{\range_{b}\cup e| e\in \GSet\setminus \range_{b}}$
and recover the set (or union of sets) in $\sS_A$ that is disjoint
from $\range_{b}$. By a simple \emph{pruning step} we can detect the ones
that are union of more than one set in $\sS_A$ and only keep the sets
in $\sS_A$.

In \lemref{recovery-success}, we show that the number of queries
that Bob is required to make to recover $\sS_A$ is $O(m^{c})$ where
$c$ is a constant.

\begin{figure}
    \centerline{%
    \begin{algorithmEnv}
        \underline{\recoverRanBit{}$\Bigl.\pth{\bigl. \GSet, s }
           $:}\+\\[0.1cm]%
        $\sS_a \leftarrow \emptyset$ \\
        \For $i=1$ to $m^{c}\log m$ \Do
        \\
        \> {\bf Let} $\range_{b}$ be a random subset of $\GSet$ of
        size $c_1\log m$
        \\
        \>\If $\existsDisj{}(s,\range_b)= {\tt true}$ %
        \\%
        \INP{%
           \>\CodeComment{// Discovering the set (or union of sets)}\\
           \>\CodeComment{// in $\sS_A$ disjoint from $\range_b$} \\}
        \>\> $\range \leftarrow \emptyset$
        \\
        \>\> \For $e \in \GSet\setminus \range_b$
        \\
        \>\>\> \If $\existsDisj(\range_b \cup e,s)={\tt false}$
        \\
        \>\>\>\>\; $\range \leftarrow \range \cup e$
        \\
        \>\> \If $\exists \range'\in \sS_a$ s.t.
        $\range \subset \range'$ \quad \CodeComment{// Pruning step}
        \\
        \>\>\> $\sS_a \leftarrow \sS_a \setminus \brc{\range'}$,
        \qquad $\sS_a \leftarrow \sS_a \cup \brc{\range}$
        \\
        \>\> \Else \If $\nexists \range'\in \sS_a$ s.t.
        $\range' \subset \range$
        \\
        \>\>\> $\sS_a \leftarrow \sS_a \cup \brc{\range}$
        \\
        \Return $\sS_a$
    \end{algorithmEnv}%
    }
    \caption{{\recoverRanBit uses a protocol for
          \prob{(Many vs One)-Set Disjointness$(m,n)$} to recover
          Alice's sets, $\sS_A$ in Bob's side.}}
    \figlab{recover-random-set-alg}
\end{figure}

\begin{lemma}
    \lemlab{uniqueness} %
    Let $\range_{b}$ be a random subset of $\GSet$ of size $c\log m$ and let
    $\sS_A$ be a collection of m random subsets of $\GSet$. The
    probability that there exists exactly one set in $\sS_A$ that is
    disjoint from $\range_{b}$ is at least ${1 \over m^{c+1}}$.
\end{lemma}
\begin{proof}
    The probability that $\range_{b}$ is disjoint from exactly one set in
    $\sS_A$ is
    \begin{align*}
       \pr(\range_{b} \text{ is disjoint from $\geq 1$ set in } \sS_A)
      - \pr(\range_{b} \text{ is disjoint from $\geq 2$ sets in }
        \sS_A)  &\geq ({1 \over 2})^{c\log m} - {m \choose 2} ({1 \over
           2})^{2c\log m}\\ 
           &\geq {1 \over m^{c+1}}.
    \end{align*}
    First we prove the first term in the above inequality. For an arbitrary set $\range \in \sS_A$,
    since any element is contained in $\range$ with probability
    ${1\over 2}$, the probability that $\range$ is disjoint from $\range_{b}$ is
    $(1/2)^{c\log m}$.  
    \begin{align*}
        \pr(\range_{b} \text{ is disjoint from at least {\bf one} set in }
        \sS_A) 
        \geq 2^{-c\log m}.
    \end{align*}
    Moreover since there exist $m \choose 2$ pairs
    of sets in $\sS_A$, and for each $\range_1,\range_2 \in \sS_A$,
    the probability that $\range_{1}$ and $\range_{2}$ are disjoint from $\range_{b}$ is
    $m^{-2c}$,
    \begin{align*}
      \pr(\range_{b} \text{ \INP{is  }disjoint from at least 
      \INP{\bf two}\InPODS{2} sets in }
      \sS_A) 
              &\leq m^{-(2c-2)}.
    \end{align*}
    
\end{proof}

A family of sets $\sM$ is called \emph{intersecting} if and only if
for any sets $A,B \in \sM$ either both $A\setminus B$ and
$B\setminus A$ are non-empty or both $A\setminus B$ and $B\setminus A$
are empty; in other words, there exists no $A, B \in \sM$ such that
$A \subseteq B$.  Let $\sS_A$ be a collection of subsets of
$\GSet$. We show that with high probability after testing $O(m^{c})$
queries for sufficiently large constant $c$, the \recoverRanBit
algorithm recovers $\sS_A$ completely if $\sS_A$ is
intersecting. First we show that with high probability the collection
$\sS_A$ is intersecting.
\begin{observation}%
    \obslab{obr:different-sets}%
    Let $\sS_A$ be a collection of $m$ uniformly random subsets of
    $\GSet$ where $\card{\GSet} \geq c\log m$. With probability at least
    $1-m^{-c/4+2}$, $\sS_A$ is an intersecting family.
\end{observation}
\begin{proof}
    The probability that $\range_1 \subseteq \range_2 $ is $({3\over 4})^n$ and
    there are at most $m(m-1)$ pairs of sets in $\sS_A$. Thus with probability at least
    $1- m^2({3 \over 4})^n \geq 1- 1/m^{{c\over 4}-2}$, $\sS_A$
    is intersecting.
\end{proof}
\begin{observation}%
    \obslab{obr:distinct-input}%
    The number of distinct inputs of Alice (collections of random subsets of $\GSet$), that is distinguishable by \recoverRanBit is 	$\Omega(2^{mn})$.
\end{observation}
\begin{proof}
    There are $2^{mn}$ collections of $m$ random subsets of $\GSet$. By \obsref{obr:different-sets}, $\Omega(2^{mn})$
    of them are intersecting. Since we can only recover the sets in the input collection and not their order, the distinct number of 
    input collection that are distinguished by \recoverRanBit is $\Omega({2^{mn} \over m!})$ which is $\Omega(2^{mn})$ 	for $n\geq c\log m$.  
\end{proof}

\noindent
By \obsref{obr:different-sets} and only considering the case such that
$\sS_A$ is intersecting, we have the following lemma.
\begin{lemma}
    \lemlab{recovery-success} %
    Let $\sS_A$ be a collection of $m$ uniformly random subsets of
    $\GSet$ and suppose that $\card{\GSet}\geq c\log m$. After testing at
    most $m^{c}$ queries, with probability at least
    $(1-{1\over m})p^{m^{c}}$, $\sS_A$ is fully recovered, where $p$
    is the success rate of protocol $\mathbf{I}$ for the \prob{(Many vs
       One)-Set Disjointness} problem.
\end{lemma}
\begin{proof}
    By \lemref{uniqueness}, for each $\range_{b}\subset\GSet$ of size
    $c_1\log m$ the probability that $\range_{b}$ is disjoint from exactly
    one set in a random collection of sets $\sS_A$ is at least
    $1/m^{c_1+1}$. Given $\range_{b}$ is disjoint from exactly one set in
    $\sS_A$, due to symmetry of the problem, the chance that $\range_{b}$ is
    disjoint from a specific set $\range \in \sS_A$ is at least
    ${1 \over m^{c_1+2}}$.  After $\alpha m^{c_1+2}\log m$
    queries where $\alpha$ is a large enough constant, for any
    $\range \in \sS_A$, the probability that there is not a query $\range_{b}$
    that is only disjoint from $\range$ is at most
    $(1-{1\over m^{c_1+2}})^{\alpha
           m^{c_1+2} \log m} 
           \leq e^{-\alpha \log m} =
        {1\over m^\alpha}$.

    Thus after trying $\alpha m^{c_1+2}\log m$ queries, with
    probability at least
    $(1- {1 \over 2m^{\alpha-1}}) \geq (1-{1\over m})$, for each
    $\range \in \sS_{A}$ we have at least one query that is only disjoint
    from $\range$ (and not any other sets in $\sS_A\setminus \range$).
    
    Once we have a query subset $\range_{b}$ which is only disjoint from a
    single set $\range \in \sS_A$, we can ask $n-c\log m$ queries of size
    $c_1\log m +1$ and recover $\range$.  Note that if $\range_{b}$ is disjoint
    from more than one sets in $\sS_A$ simultaneously, the process
    (asking $n-c\log m$ queries of size $c_1\log m +1$) will end up in
    recovering the union of those sets. Since $\sS_A$ is an
    intersecting family with high probability
    (\obsref{obr:different-sets}), by \emph{pruning step} in the
    \recoverRanBit algorithm we are guaranteed that at the end of
    the algorithm, what we returned is exactly $\sS_A$.  Moreover the
    total number of queries the algorithm makes is at most
    \begin{align*}
    n\times (\alpha m^{c_1+2}\log m) \leq \alpha 
    m^{c_1+3}\log m \leq m^c
    \end{align*} 
    for $c\geq c_1+4$.

    Thus after testing $m^c$ queries, $\sS_A$ will be recovered with
    probability at least $(1-{1\over m})p^{m^c}$ where $p$ is the
    success probability of the protocol $\mathbf{I}$ for \prob{(Many
       vs One)-Set Disjointness($m,n$)}.
\end{proof}

\begin{corollary}%
    \corlab{cor:constant-pass-protocol}%
    Let $\textbf{I}$ be a protocol for \prob{(Many vs One)-Set
       Disjointness($m,n$)} with error probability $O(m^{-c})$ and $s$
    bits of communication such that $n \geq c\log m$ for large enough
    $c$. Then \recoverRanBit recovers $\sS_A$ with
    constant success probability using $s$ bits of communication.
\end{corollary}

By \obsref{obr:distinct-input}, since \recoverRanBit distinguishes 
$\Omega(2^{mn})$ distinct inputs with
constant probability of success (by
\corref{cor:constant-pass-protocol}), the size of message sent by
Alice, should be $\Omega(mn)$. This proves
\thmref{rand-all-one-set-disjointness}.

\begin{proofof}{\thmref{main}:}
    As we showed earlier, the communication complexity of \prob{(Many
    vs One)-Set Disjointness} is a lower bound for the
    communication complexity of \prob{Set Cover}. 
    \thmref{rand-all-one-set-disjointness} showed that any
    protocol for \prob{(Many vs One)-Set Disjointness$(n,\card{\sS_A)}$} with
    error probability less than $O(m^{-c})$ requires $\Omega(mn)$ bits
    of communication.  Thus any single-round randomized protocol for
    \probSetCover with error probability $O(m^{-c})$ requires
    $\Omega(mn)$ bits of communication.
\end{proofof}

Since any $p$-pass streaming $\alpha$-approximation algorithm for
problem \prob{P} that uses $O(s)$ memory space, is a $p$-round
two-party $\alpha$-approximation protocol for problem \prob{P} using
$O(sp)$ bits of communication \cite{gm-tlbsp-08}, and by
\thmref{main}, we have the following
lower bound for \probSetCover problem in the streaming model.

\begin{theorem}%
    \thmlab{lower-bound-set-cover}%
    Any single-pass randomized streaming algorithm for \prob{Set
       Cover($\GSet,\sS$)} that computes a $(3/2)$-approximate
    solution with probability $\Omega(1-m^{-c})$ requires $\Omega(mn)$
    memory space (assuming $n\geq c_1\log m$).
\end{theorem}

\section{Geometric Set Cover}
\seclab{geometric-set-cover}
In this section, we consider the streaming \probSetCover problem
in the geometric settings.  We present an algorithm for the case where
the elements are a set of $n$ points in the plane $\Re^2$ and
the $m$ sets are either all disks, all axis-parallel rectangles, or
all $\alpha$-fat triangles (which for simplicity we call shapes) given
in a data stream. As before, the goal is to find the minimum size
cover of points from the given sets. We call this problem the
\prob{Points-Shapes Set Cover} problem.

Note that, the description of each shape requires $O(1)$ space and
thus the \prob{Points-Shapes Set Cover} problem is trivial to be
solved in $O(m+n)$ space. In this setting the goal is to design an
algorithm whose space is sub-linear in $O(m+n)$.
Here we show that almost the same algorithm as
\algIterSC (with slight modifications) uses
$\tldO(n)$ space to find an $O(\rho)$-approximate solution
of the \prob{Points-Shapes Set Cover} problem in \emph{constant} passes.
\subsection{Preliminaries}
A triangle $\triangle$ is called \emph{$\alpha$-fat} (or simply fat) if the ratio
between its longest edge and its height on this edge is bounded by a
constant $\alpha > 1$ (there are several equivalent definitions of
$\alpha$-fat triangles).
\begin{definition}
	\deflab{canonical}
    Let $(\GSet,\sS)$ be a set system such that $\GSet$ is a set of points
    and $\sS$ is a collection of shapes, in the plane
    $\Re^2$. The \emph{canonical representation} of $(\GSet,\sS)$
    is a collection $\sS'$ of regions such that the following conditions
    hold. First, each $\range'\in \sS'$ has $O(1)$ description. Second, for
    each $\range'\in \sS'$, there exists $\range\in \sS$ such that
    $\range'\cap \GSet \subseteq \range\cap \GSet$. Finally, for each $\range\in \sS$,
    there exists $c_1$ sets $\range'_1,\cdots,\range'_{c_1} \in \sS'$ such that
    $\range\cap \GSet = (\range'_1\cup\cdots\cup \range'_{c_1})\cap \GSet$ for some
    constant $c_1$.
\end{definition}

The following two results are from \cite{ehr-gpnuc-12} which are the formalization of the ideas in \cite{aes-ssena-10}.
\begin{lemma}
    \lemlab{rect}%
    {\em(Lemma~4.18 in \cite{ehr-gpnuc-12})} %
    Given a set of points $\GSet$ in the plane $\Re^2$ and a
    parameter $w$, one can compute a set $\sS'_{\total}$ of
    $O(\card{\GSet}w^2\log \card{\GSet})$ axis-parallel rectangles with
    the following property. For an arbitrary axis-parallel rectangle
    $\range$ that contains at most $w$ points of $\GSet$, there exist two
    axis-parallel rectangles $\range'_1,\range'_2\in \sS'_{\total}$ whose union
    has the same intersection with $\GSet$ as $\range$, i.e.,
    $\range\cap \GSet = (\range_1'\cup \range_2') \cap \GSet$.
\end{lemma}
\begin{lemma}
    \lemlab{fat}%
    {\em(Theorem 5.6 in \cite{ehr-gpnuc-12})} %
    Given a set of points $\GSet$ in $\Re^2$, a
    parameter $w$ and a constant $\alpha$, one can compute a set
    $\sS'_{\total}$ of $O(\card{\GSet}w^3\log^2 \card{\GSet})$ regions
    each having $O(1)$ description with the following property. For an
    arbitrary $\alpha$-fat triangle $\range$ that contains at most $w$
    points of $\GSet$, there exist nine regions from $\sS'_{\total}$
    whose union has the same intersection with $\GSet$ as $\range$.
\end{lemma}
Using the above lemmas we get the following lemma.
\begin{lemma}
    \lemlab{canonical} %
    Let $\GSet$ be a set of points in $\Re^2$ and let $\sS$ be a
    set of shapes (discs, axis-parallel rectangles or fat triangles),
    such that each set in $\sS$ contains at most $w$ points of
    $\GSet$. Then, in a single pass over the stream of sets $\sS$, one
    can compute the canonical representation $\sS'$ of
    $(\GSet,\sS)$. Moreover, the size of the canonical representation is at most
    $O(\card{\GSet}w^3\log^2 \card{\GSet})$ and the space requirement of
    the algorithm is
    $\tldO(\card{\sS'}) = \tldO(\card{\GSet}w^3)$.
\end{lemma}
\begin{proof}
    For the case of axis-parallel rectangles and fat triangles, first
    we use \lemref{rect} and \lemref{fat} to get the set
    $\sS'_{\total}$ offline which require
    $\tldO(\sS'_{\total}) = \tldO(\card{\GSet}w^3\log^2 \card{\GSet})$ memory
    space. Then by making one pass over the stream of sets $\sS$, we can
    find the canonical representation $\sS'$ by picking all the sets
    $S'\in \sS'_{\total}$ such that
    $S'\cap \GSet \subseteq S\cap\GSet$ for some $S\in \sS$.  For
    discs however, we just make one pass over the sets $\sS$ and keep
    a maximal subset $\sS'\subseteq \sS$ such that for each pair of
    sets $S_1',S_2'\in \sS'$ their projection on $\GSet$ are
    different, i.e., $S_1'\cap \GSet \neq S_2'\cap \GSet$.  By a
    standard technique of Clarkson and Shor \cite{cs-arscg-89}, it can
    be proved that the size of the canonical representation, i.e.,
    $\card{S'}$, is bounded by $O(\card{\GSet} w^2)$. Note that
    this is just counting the number of discs that contain at most $w$ 
    points, namely the at most $w$-level discs.
\end{proof}
\subsection{Algorithm}
 
The outline of the \prob{Points-Shapes-Set-Cover} algorithm (shown in
\figref{shapes-alg}) is very similar to the \algIterSC algorithm
presented earlier in \secref{unweighted-set-cover}.
 
In the first pass, the algorithm picks all the sets that cover a large
number of yet-uncovered elements.  Next, we sample $\GSample$. Since
we have removed all the ranges that have large size, in the first
pass, the size of the remaining ranges \emph{restricted to} the sample
$\GSample$ is small. Therefore by \lemref{canonical}, the canonical
representation of $(\GSample,\sSSample)$ has small size and we can
afford to store it in the memory. We use \lemref{canonical} to compute
the canonical representation $\sSSample$ in one pass.  The algorithm
then uses the sets in $\sSSample$ to find a cover $\solS$ for the
points of $\GSample$. Next, in one additional pass, the algorithm
replaces each set in $\solS$ by one of its supersets in $\sS$.
 
Finally, note that in the algorithm of \secref{unweighted-set-cover},
we are assuming that the size of the optimal solution is $O(k)$. Thus
it is enough to stop the iterations once the number of uncovered
elements is less than $k$. Then we can pick an arbitrary set for
each of the uncovered elements. This would add only $k$ more sets to
the solution. Using this idea, we can reduce the size of the sampled
elements down to $c\rho k({n\over k})^{\delta}\log m\log n$ which would help
us in getting near-linear space in the geometric setting. Note that the final pass of the
algorithm can be embedded into the previous passes but for the
sake of clarity we write it separately.

\newcommand{\algGSC}{\Algorithm{algGeomSC}\xspace}%

\begin{figure}[t]
    \centerline{%
       \begin{algorithmEnvX}
           \underline{\algGSC{}$\Bigl.\pth{\bigl. \GSet, \sS, \delta }
              $:}\+\\[0.1cm]%
           \For $k \in \set{2^{i}\sep 0\leq i \leq \log n}$ \Do
           \textbf{in parallel}: \INP{\CodeComment{//
                 $n =\card{\GSet}$}}%
           \\%
           \> {\bf Let} $\GLeftover \leftarrow \GSet$ and
           $\sol \leftarrow \emptyset$ \\%
           \> {\bf Repeat} $1/ \delta$ times: \\%
           \>\> \For $\range \in \sS$ \Do \quad\CodeComment{// Pass}
           \\%
           \>\>\> \If
           $\card{\range \cap \GLeftover} \geq \card{\GSet}/k$ \Then
           \\%
           \>\>\>\> $\sol \leftarrow \sol \cup \brc{\range}$%
           \\\>\>\>\>
           $\GLeftover \leftarrow \GLeftover\setminus \range$
           \\[0.1cm]%
           \>\> $\GSample \leftarrow$ sample of $\GLeftover$ of size
           $c\rho k(n/k)^{\delta} \log m \log n$ \\%
           \>\>$\sSSample \leftarrow $%
           \Algorithm{compCanonicalRep}{}%
           $(\GSample,\sS,{\card{\GSample} \over k})$ 
           \CodeComment{// Pass} \\%
           \>\> $\solS \leftarrow $ \offlineSC{}%
           ($\GSample,\sSSample$) \\%
           \>\> \For $\range \in \sS$ \Do \quad\CodeComment{// Pass} \\
           \>\>\> \If $\exists \range' \in \solS$ s.t.
           $\range' \cap \GSample \subseteq \range \cap \GSample$
           \Then\\
           \>\>\>\> $\sol \leftarrow \sol \cup \brc{\range}$%
           \\\>\>\>\>
           $\solS \leftarrow \solS \setminus \brc{\range'}$\\%
           \>\>\>\> $\GLeftover \leftarrow \GLeftover\setminus \range$\\
           \>\For $\range \in \sS$ \Do \quad\CodeComment{// Final
              Pass} \\%
           \>\> \If $\range \cap \GLeftover \neq \emptyset$ \Then \\%
           \>\>\> $\sol \leftarrow \sol \cup \brc{\range}$%
           \\\>\>\>
           $\GLeftover \leftarrow \GLeftover \setminus \range$ \\
           \Return smallest $\sol$ computed in parallel
       \end{algorithmEnvX}%
    }
    \caption{{A streaming algorithm for \prob{Points-Shapes Set Cover}
       problem.}}
    \figlab{shapes-alg}
\end{figure}
\subsection{Analysis}
By a similar approach to what we used in \secref{unweighted-set-cover}
to analyze the pass count and approximation guarantee of \algIterSC
algorithm, we can show that the number of passes of the \algGSC{}
algorithm is ${3/\delta}+1$ (which can be reduced to $3/\delta$ with
minor changes), and the algorithm returns an
$O(\rho/\delta)$-approximate solution. Next, we analyze the space
usage and the correctness of the algorithm. Note that our analysis in
this section only works for $\delta \leq 1/4$.
\begin{lemma}\lemlab{geometric-space}
    The algorithm uses $\tldO(n)$ space.
\end{lemma}
\begin{proof}
    Consider an iteration of the algorithm. The memory space used in
    the first pass of each iteration is $\tldO(n)$. The size of
    $\GSample$ is $c\rho k (n/k)^{\delta} \log m \log n$ and after
    the first pass the size of each set is at most
    $\card{\GSet}/k$. Thus using Chernoff bound for each set
    $\range\in \sS\setminus \sol$,
    \begin{align*}
        \Prob{\card{\range\cap \GSample} > (1+ 2) {\card{\GSet} \over
              k}\times {\card{\GSample} \over \card{\GSet}}}%
        &\leq%
        \exp\pth{-{4\card{\GSample} \over 3k}}%
        \leq %
        ({1\over m})^{c+1}.
    \end{align*}
    Thus, with probability at least $1-m^{-c}$ (by the union bound),
    all the sets that are not picked in the first pass, cover at most
    $3\card{\GSample}/k = c\rho(n/k)^\delta \log m \log n$ elements of
    $\GSample$. Therefore, we can use \lemref{canonical} to show that
    the number of sets in the canonical representation of
    $(\GSample,\sSSample)$ is at most 
    \begin{align*}
    O(\card{\GSample} \pth{{3\card{\GSample}\over k}}^{\!3}
           \log^2 \card{\GSample})= O(\rho^4 n \log^4 m \log^6 n),
	\end{align*}
    as long as $\delta \leq 1/4$.  To store each set in a canonical
    representation of $(\GSample,\sS)$ only constant space is
    required. Moreover, by \lemref{canonical}, the space requirement of the second pass is $\tldO (\card{\sSSample}) = \tldO(n)$. Therefore, the total required space is $\tldO(n)$ and
    the lemma follows.
\end{proof}

\begin{theorem}\thmlab{geometric-set-cover}
    Given a set system defined over a set $\GSet$ of $n$ points in the
    plane, and a set of $m$ ranges $\sS$ (which are either all disks,
    axis-parallel rectangles, or fat triangles). Let $\rho$ be the
    quality of approximation to the offline set-cover solver we have,
    and let $0 < \delta <1/4$ be an arbitrary parameter.
    
    Setting $\delta=1/4$, the algorithm \algGSC, depicted in \figref{shapes-alg}, with high
    probability, returns an $O(\rho)$-approximate solution of
    the optimal set cover solution for the instance $(\GSet,
    \sS)$.
    This algorithm uses $\tldO(n)$ space, and performs constant
    passes over the data.
\end{theorem}
\begin{proof}
    As before consider the run of the algorithm in which
    $\card{\opt} \leq k < 2\card{\opt}$. Let $\GSetA$ be the set of 
	uncovered elements $\GLeftover$ at the
    beginning of the iteration and note that the total number of sets
    that is picked during the iteration is at most $(1+c_1\rho)k$ where $c_1$ is the constant defined in \defref{canonical}. Let $\FF$ denote 
	all possible
	such covers, that is $\FF = \Set{ \sS' \subseteq \sS}{\bigl.\card{\sS'}\leq
       (1+c_1\rho)k}$. 
    Let $\sSA$ be the collection that contains all possible set of uncovered 
	elements at the end of the iteration, defined as
    \begin{math}
		\sSA=\Set{\GSetA\setminus \bigcup_{\range \in \Cover}
			 \range}{ \Cover \in \FF}
	\end{math}.
    Set $p = ({k/n})^{\delta}$,
    $\eps=1/2$ and $q=m^{-c}$. Since for large enough $c$,
    ${c'\over \eps^2 p}(\log \card{\sSA} \log {1\over p} +
    \log{1\over q}) \leq c\rho k(n/k)^\delta\log m \log n =
    \card{\GSample}$
    with probability at least $1-m^{-c}$, by \lemref{relative-approx}, the 
	set of sampled elements $\GSample$ is a relative $(p,\eps)$-approximation 
	sample of $(\GSetA, \sSA)$.
 
    Let $\Cover \subseteq \sS$ be the collection of sets picked in the third
    pass of the algorithm that covers all elements in $\GSample$. By
    \lemref{canonical}, $\card{\Cover} \leq c_1\rho k$ for some
    constant $c_1$.  Since with high probability $\GSample$ is a
    relative $(p,\eps)$-approximation sample of $(\GSetA, \sSA)$, the number of
    uncovered elements of $\GSetA$ (or $\GLeftover$) after adding $\Cover$ to
    $\sol$ is at most $\eps p|\GSetA| \leq {\card{\GSet}(k/n)^{\delta}}$.
    Thus with probability at least $(1-m^{-c})$, in each iteration and
    by adding $O(\rho k)$ sets, the number of uncovered elements
    reduces by a factor of $(n/k)^{\delta}$.  
    
    Therefore, after $4$ iterations (for $\delta =1/4$) the algorithm 
    picks $O(\rho k)$ sets and with high probability the number of uncovered 
    elements is at most $n (k/n)^{\delta/\delta} = k$.  Thus, in the final pass the
    algorithm only adds $k$ sets to the solution $\sol$, and hence the
    approximation factor of the algorithm is $O(\rho)$.
\end{proof}

\begin{remark}
    The result of \thmref{geometric-set-cover} is similar to the
    result of Agarwal and Pan \cite{ap-nlagh-14} -- except that their
    algorithm performs $O( \log n)$ iterations over the data, while
    the algorithm of \thmref{geometric-set-cover} performs only a
    constant number of iterations. In particular, one can use the
    algorithm of Agarwal and Pan \cite{ap-nlagh-14} as the offline
    solver.
\end{remark}


\section{Lower bound for multipass \INP{algorithms}}
\seclab{lower-bound-multipass}

In this section we give lower bound on the memory space of multipass
streaming algorithms for the \probSetCover problem. Our main result
is $\Omega(mn^{\delta})$ space for streaming algorithms that return an
optimal solution of the \probSetCover problem in $O(1/\delta)$
passes for $m=O(n)$. Our approach is to reduce the communication \prob{Intersection
   Set Chasing($n,p$)} problem introduced by Guruswami and Onak
\cite{go-slbmg-13} to the communication \probSetCover problem.

Consider a communication problem $\mathsf{P}$ with $n$ players
$P_1,$ $\cdots, P_n$. The problem $\mathsf{P}$ is a $(n,r)$-communication
problem if players communicate in $r$ rounds and in each round they
speak in order $P_1,\cdots, P_n$. At the end of the $r$\th round $P_n$
should return the solution. Moreover we assume private randomness and
public messages.  In what follows we define the communication
\prob{Set Chasing} and \prob{Intersection Set Chasing} problems.
\begin{definition}[Communication \prob{Set Chasing} Problem] The 
	\prob{Set Chasing($n,p$)} problem
	is a $(p,p-1)$ communication problem in
    which the player $i$ has a function $f_i: [n] \rightarrow 2^{[n]}$
    and the goal is to compute
    $\vec{f_1}(\vec{f_2}(\cdots\vec{f_p}(\set{1})\cdots))$ where
    $\vec{f_i}(S)= \bigcup_{s\in S} f_{i}(s)$.
    \figref{set-chasing}(a) shows an instance of the communication
    \prob{Set Chasing}($4,3$).
\end{definition}
\begin{definition}[Communication \prob{Intersection Set
       Chasing}] The \prob{Intersection
       Set Chasing($n,p$)} is a $(2p,p-1)$ communication problem in
    which the first $p$ players have an instance of the \prob{Set
       Chasing($n,p$)} problem and the other $p$ players have another
    instance of the \prob{Set Chasing($n,p$)} problem. The output of
    the \prob{Intersection Set Chasing($n,p$)} is $1$ if the solutions
    of the two instances of the \prob{Set Chasing($n,p$)} intersect
    and $0$ otherwise.  \figref{set-chasing}(b) shows an instance
    of the \prob{Intersection Set Chasing}($4,3$). The function $f_i$
    of each player $P_i$ is specified by a set of directed edges form
    a copy of vertices labeled $\set{1, \cdots, n}$ to another copy of
    vertices labeled $\set{1, \cdots, n}$.
\end{definition}
\begin{figure}
    \centering
    \includegraphics[height=1.5in]{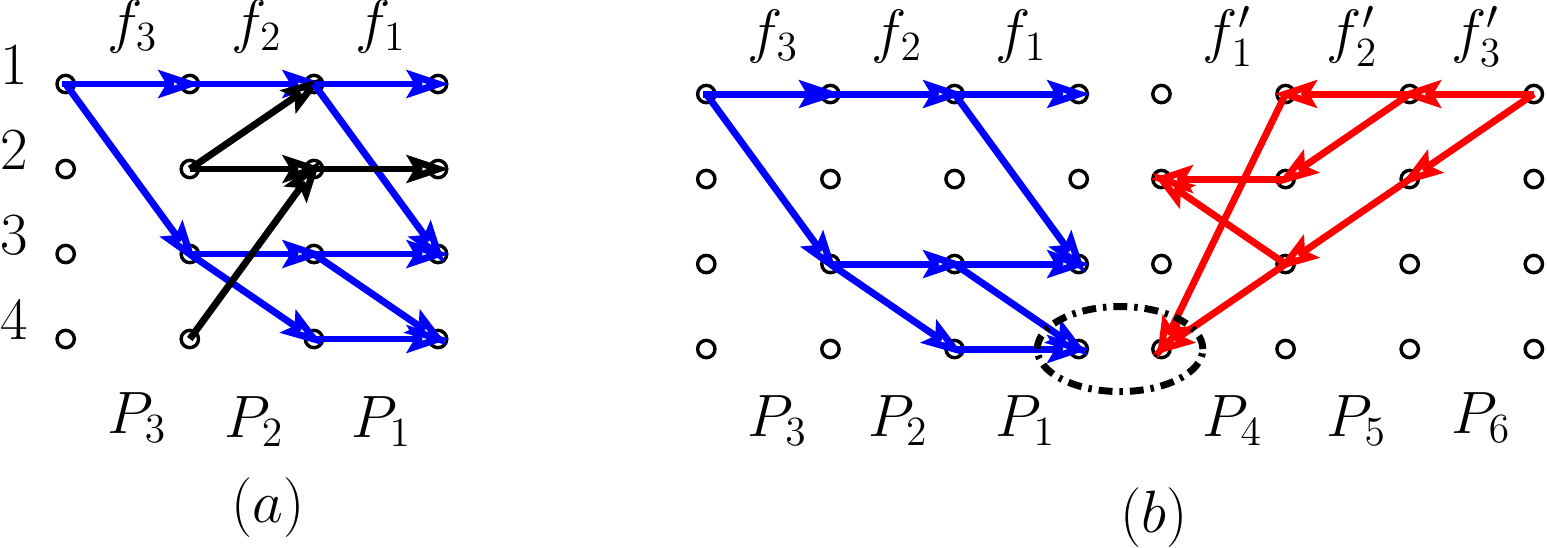}
    \caption{{(a) shows an example of the communication
          \prob{Set Chasing($4,3$)} and (b) is an instance of the
          communication \prob{Intersection Set Chasing($4,3$)}.}}
    \figlab{set-chasing}
    \InPODS{\vspace{-0.3cm}}%
\end{figure}

The communication \prob{Set Chasing} problem is a generalization of
the well-known communication \prob{Pointer Chasing} problem in which
player $i$ has a function $f_i:[n] \rightarrow [n]$ and the goal is to
compute $f_1(f_2(\cdots f_p(1)\cdots))$.

\cite{go-slbmg-13} showed that any randomized protocol that solves
\prob{Intersection Set Chasing($n,p$)} with error probability less
than $1/10$, requires
$\Omega({n^{1+1/(2p)} \over p^{16} \log^{3/2}n})$ bits of
communication where $n$ is sufficiently large and
$p \leq {\log n \over \log \log n}$.  In \thmref{reduction}, we
reduce the communication \prob{Intersection Set Chasing} problem to
the communication \probSetCover problem and then give the first
superlinear memory lower bound for the streaming \probSetCover
problem.

\begin{figure}[t]
    \centering
    \includegraphics[height=3.6in]{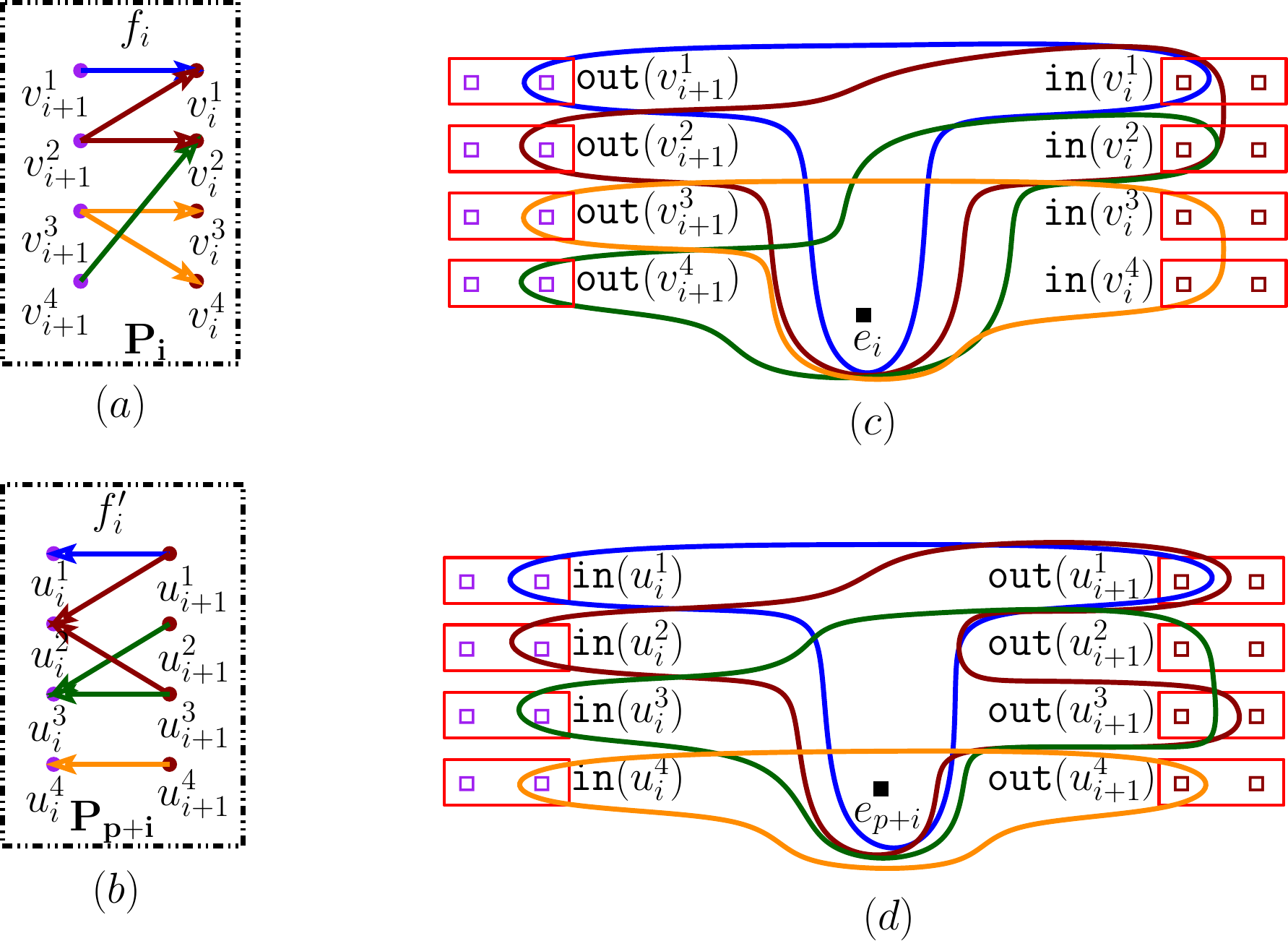}
    \caption{{The gadgets used in the reduction of the
          communication \prob{Intersection Set Chasing} problem to the
          communication \probSetCover problem. (a) and (c) shows the
          construction of the gadget for players $1$ to $p$ and
          (c) and (d) shows the construction of the gadget for players
          $p+1$ to $2p$.}}
    \figlab{reduction}%
    \InPODS{\vspace{-0.6cm}}%
\end{figure}

\begin{definition}[Communication \prob{Set Cover($\GSet, \sS, p$)}
    Problem]
     
    The communication \prob{Set Cover($n,p$)} is a $(p,p-1)$
    communication problem in which a collection of elements $\GSet$ is
    given to all players and each player $i$ has a collection of
    subsets of $\GSet$, $\sS_i$. The goal is to solve \prob{Set
       Cover($\GSet, \sS_1\cup \cdots \cup \sS_p$)} using the minimum
    number of communication bits.
\end{definition}

\begin{theorem}
    \thmlab{reduction}%
    Any $({1 / 2\delta}-1)$ passes streaming algorithm that   
    solves the \prob{Set Cover($\GSet,\sS$)}
    optimally with constant probability of error
      requires $\tldOmega(mn^{\delta})$
    memory space where $\delta \geq {\log \log n \over \log n}$ and $m=O(n)$.
\end{theorem}
Consider an instance $\mathsf{ISC}$ of the communication
\prob{Intersection Set Chasing($n,p$)}. We construct an instance of
the communication \prob{Set Cover($\GSet,\sS, 2p$)} problem such that
solving \prob{Set Cover($\GSet,\sS$)} optimally determines whether the
output of $\mathsf{ISC}$ is $1$ or not.

The instance $\mathsf{ISC}$ consists of $2p$ players. Each player
$1, \cdots, p$ has a function $f_i:[n] \rightarrow 2^{[n]}$ and each
player $p+1, \cdots, 2p$ has a function $f'_i:[n] \rightarrow 2^{[n]}$
(see \figref{set-chasing}). In $\mathsf{ISC}$, each function $f_i$ is
shown by a set of vertices $v_i^1, \cdots, v_i^n$ and
$v_{i+1}^1, \cdots, v_{i+1}^n$ such that there is a directed edge from
$v_{i+1}^{j}$ to $v_{i}^\ell$ if and only if $\ell \in
f_i(j)$. Similarly, each function $f'_i$ is denoted by a set of
vertices $u_i^1, \cdots, u_i^n$ and $u_{i+1}^1, \cdots, u_{i+1}^n$
such that there is a directed edge from $u_{i+1}^{j}$ to $u_{i}^\ell$
if and only if $\ell \in f'_i(j)$ (see \figref{reduction}(a) and
\figref{reduction}(b)).

In the corresponding communication \probSetCover instance of
$\mathsf{ISC}$, we add two elements ${\tt in}(v_{i}^j)$ and
${\tt out}(v_{i}^j)$ per each vertex $v_i^{j}$ where
$i\leq p+1, j\leq n$. We also add two elements ${\tt in}(u_{i}^j)$ and
${\tt out}(u_{i}^j)$ per each vertex $u_i^{j}$ where
$i\leq p+1, j\leq n$. In addition to these elements, for each player
$i$, we add an element $e_{i}$ (see \figref{reduction}(c) and
\figref{reduction}(d)).

Next, we define a collection of sets in the corresponding \prob{Set
   Cover} instance of $\mathsf{ISC}$. For each player $P_i$, where
$1\leq i\leq p$, we add a single set $S_{i}^j$ containing
${\tt out}(v_{i+1}^j)$ and ${\tt in}(v_i^\ell)$ for all out-going
edges $(v_{i+1}^j, v_i^\ell)$. Moreover, all $S_{i}^j$ sets contain
the element $e_i$.  Next, for each vertex $v_i^j$ we add a set $R_i^j$
that contains the two corresponding elements of $v_i^j$,
${\tt in}(v_{i}^j)$ and ${\tt out}(v_{i}^j)$. In
\figref{reduction}(c), the red rectangles denote $R$-type sets and
the curves denote $S$-type sets for the first half of the players.

Similarly to the sets corresponding to players $1$ to $p$, for each
player $P_{p+i}$ where $1\leq i \leq p$, we add a set $S_{p+i}^j$
containing ${\tt in}(u_i^j)$ and ${\tt out}(u_{i+1}^\ell)$ for all
in-coming edges $(u_{i+1}^\ell, u_{i}^j)$ of $u_i^j$ (denoting
$f'^{-1}_i(j)$).  The set $S_{p+i}^j$ contains the element $e_{p+i}$
too. Next, for each vertex $u_i^j$ we add a set $T_{p+i}^j$ that
contains the two corresponding elements of $u_i^j$,
${\tt in}(u_{i}^j)$ and ${\tt out}(u_{i}^j)$.  In
\figref{reduction}(d), the red rectangles denote $T$-type sets and
the curves denote $S$-type sets for the second half of the players.

At the end, we merge $v_1^{i}$s and $u_{1}^{i}$s as shown in
\figref{middle}. After merging the corresponding sets of $v_1^j$s
($R_1^1, \cdots, R_1^{n}$) and the corresponding sets of $u_1^j$s
($T_1^1, \cdots,$ $T_1^{n}$), we call the merged sets
$T_{1}^1, \cdots, T_{1}^n$.
\begin{figure}[htb]
    \centering
    \includegraphics[height=1.8in]{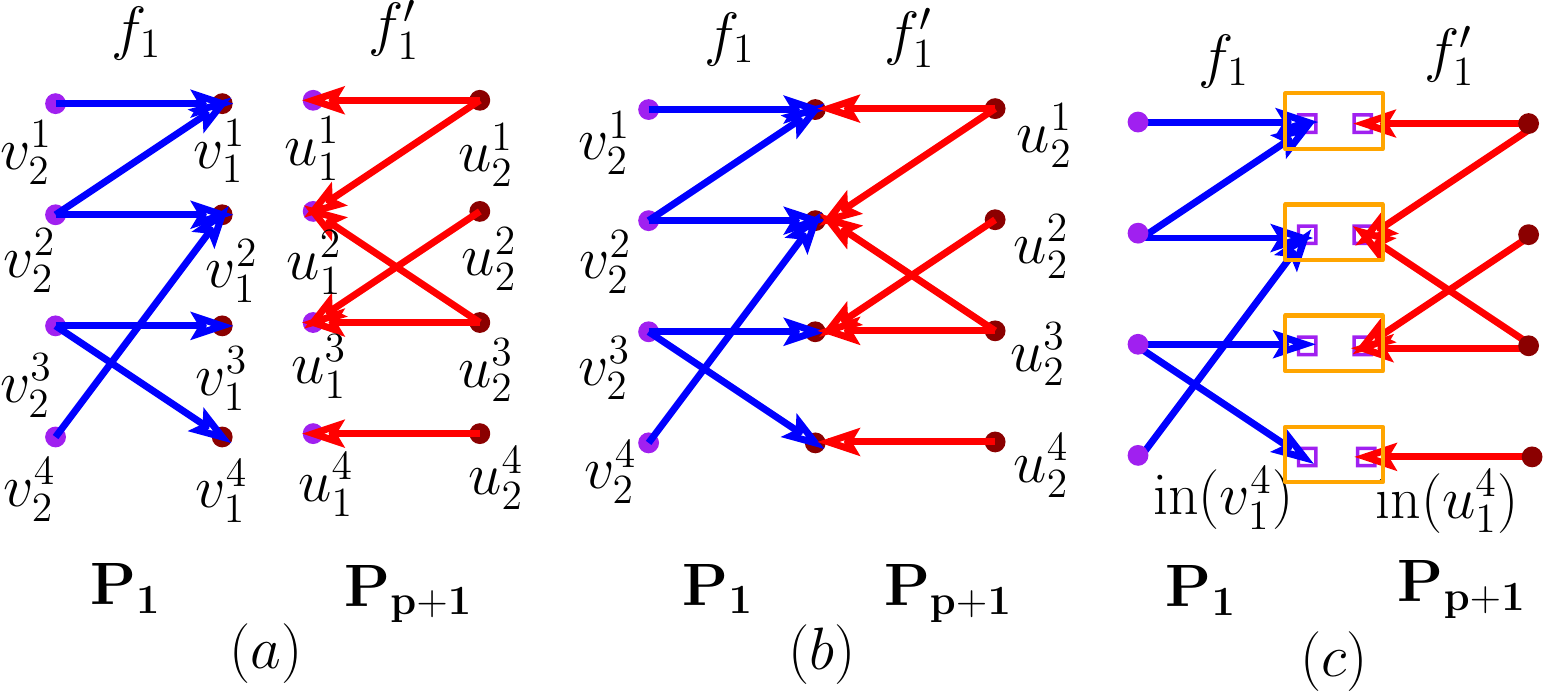}
    \caption{{In (b) two \prob{Set Chasing} instances merge in
          their first set of vertices and (c) shows the corresponding
          gadgets of these merged vertices in the communication
          \probSetCover.}}
    \figlab{middle}
\end{figure}
\begin{figure}
    \centering {\includegraphics[height=1.75in]{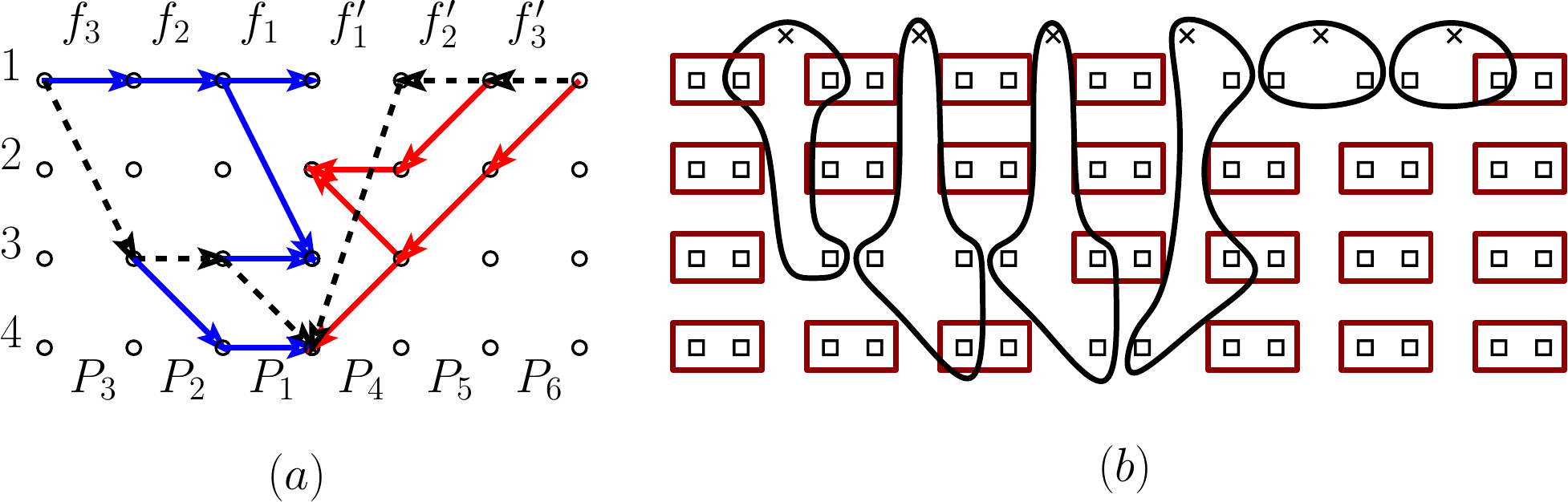}}
    \caption{{In $(a)$, path $Q$ is shown with black dashed arcs
          and $(b)$ shows the corresponding cover of path $Q$.}}
    \figlab{optimal}
\end{figure}

The main claim is that if the solution of $\mathsf{ISC}$ is $1$ then
the size of an optimal solution of its corresponding \probSetCover
instance $\mathsf{SC}$ is $(2p+1)n+1$; otherwise, it is $(2p+1)n+2$.
\begin{lemma}
    \lemlab{optimal-lower-bound} %
    The size of any feasible solution of $\mathsf{SC}$ is at least
    $(2p+1)n +1$.
\end{lemma}
\begin{proof}
    For each player $i$ ($1\leq i \leq p$), since
    ${\tt out}(v_{i+1}^{j})$s are only covered by $R_{i+1}^j$ and
    $S_{i}^j$, at least $n$ sets are required to cover
    ${\tt out}(v_{i+1}^{1}), \cdots, {\tt out}(v_{i+1}^{n})$. Moreover
    for player $P_p$, since ${\tt in}(v_{p+1}^{j})$s are only covered
    by $R_{p+1}^j$ and $e_{p}$ is only covered by $S_{p}^{1}$, all
    $n+1$ sets $R_{p+1}^1, \cdots, R_{p+1}^n, S_p^1$ must be selected
    in any feasible solution of $\mathsf{SC}$.

    Similarly for each player $p+i$ ($1 \leq i \leq p$), since
    in($u_i^j$)s are only covered by $T_{i}^j$ and $S_{p+i}^{j}$, at
    least $n$ sets are required to cover
    ${\tt in}(u_{i}^{1}), \cdots, {\tt in}(u_{i}^{n})$.  Moreover,
    considering $u_{p+1}^1, \cdots, u_{p+1}^n$, since
    ${\tt in}(u_{p+1}^{j})$ is only covered by $T_{p+1}^j$, all $n$
    sets $T_{p+1}^1, \cdots, T_{p+1}^n$ must be selected in any
    feasible solution of $\mathsf{SC}$.

    All together, at least $(2p+1)n + 1$ sets should be selected in
    any feasible solution of $\mathsf{SC}$.
\end{proof}
\begin{lemma}
    \lemlab{optimal-solution-size} %
    Suppose that the solution of $\mathsf{ISC}$ is $1$. Then the size
    of an optimal solution of its corresponding \probSetCover
    instance is exactly $(2p+1)n+1$.
\end{lemma}
\begin{proof}
    By \lemref{optimal-lower-bound}, the size of an optimal
    solution of $\mathsf{S}$ is at least $(2p+1)n+1$. Here we prove
    that $(2p+1)n+1$ sets suffice when the solution of $\mathsf{ISC}$
    is $1$.  Let
    $Q= v_{p+1}^1, v_{p}^{j_p},\dots, v_{2}^{j_2}, v_{1}^{j_1},
    u_{1}^{\ell_1}, u_{2}^{\ell_2}, \dots,u_{p}^{\ell_p}, u_{p+1}^1$
    be a path in $\mathsf{ISC}$ such that $j_1 = \ell_1$ (since the
    solution of $\mathsf{ISC}$ is $1$ such a path exists). The
    corresponding solution to $Q$ can be constructed as follows (See
    \figref{optimal}):
    \begin{itemize}
        \setlength{\itemsep}{1pt} \setlength{\parskip}{0pt}
        \setlength{\parsep}{0pt}
        \item{Pick $S_{p}^1$ and all $R_{p+1}^{j}$s ($n+1$ sets).}
        \item{For each $v_{i}^{j_i}$ in $Q$ where $1<i\leq p$, pick
           the set $S_{i-1}^{j_i}$ in the solution. Moreover, for each
           such $i$ pick all sets $R_{i}^j$ where $j\neq j_i$
           ($n(p-1)$ sets).}
        \item{For $v_{1}^{j_1}$ (or $u_{1}^{\ell_1}$), pick the set
           $S_{p+1}^{j_1}$. Moreover, pick all sets $T_{1}^j$ where
           $j\neq j_1$ ($n$ sets).}
        \item{For each $u_{i}^{\ell_i}$ in $Q$ where $1<i\leq p$, pick
           the set $S_{p+i}^{\ell_i}$ in the solution. Moreover, for
           each such $i$ pick all sets $T_{i}^\ell$ where
           $\ell \neq \ell_i$} ($n(p-1)$ sets).
        \item{ Pick all $T_{p+1}^{j}$s ($n$ sets).}
    \end{itemize}
    It is straightforward to see that the solution constructed above
    is a feasible solution.
\end{proof}
\begin{lemma}
    \lemlab{optimal-cover} %
    Suppose that the size of an optimal solution of the corresponding
    \probSetCover instance of $\mathsf{ISC}$, $\mathsf{SC}$, is
    $(2p+1)n+1$. Then the solution of $\mathsf{ISC}$ is $1$.
\end{lemma}
\begin{proof}
    As we proved earlier in \lemref{optimal-lower-bound}, any
    feasible solution of $\mathsf{SC}$ picks
    $R_{p+1}^1, \cdots, R_{p+1}^n, S_p^1$ and $T_{p+1}^1, \cdots,$
    $T_{p+1}^n$. Moreover, we proved that for each $1\leq i<p$, at
    least $n$ sets should be selected from
    $R_{i+1}^1, \cdots, R_{i+1}^n, S_i^1,$ $\cdots, S_i^n$. Similarly,
    for each $1\leq i \leq p$, at least $n$ sets should be selected
    from $T_{i}^1, \cdots, T_{i}^n$, $S_{p+i}^1, \cdots, S_{p+1}^n$.
    Thus if a feasible solution of $\mathsf{SC}$, $\opt$, is of size
    $(2p+1)n+1$, it has exactly $n$ sets from each specified group.

    Next we consider the first half of the players and second half of
    the players separately. Consider $i$ such that $1\leq i<p$. Let
    $S_i^{j_1}, \cdots, S_i^{j_k}$ be the sets picked in the optimal
    solution (because of $e_{i}$ there should be at least one set of
    form $S_i^j$ in $\opt$). Since each ${\tt out}(v_{i+1}^j)$ is only
    covered by $S_i^j$ and $R_{i+1}^j$, for all
    $j\notin\set{j_1,\dots,j_k}$, $R_{i+1}^j$ should be selected in
    $\opt$. Moreover, for all $j \in \set{j_1, \cdots, j_k}$,
    $R_{i+1}^j$ should not be contained in $\opt$ (otherwise the size
    of $\opt$ would be larger than $(2p+1)n +1$). Consider
    $j\in \set{j_1,\dots, j_k}$. Since $R_{i+1}^j$ is not in $\opt$,
    there should be a set $S_{i+1}^{\ell}$ selected in $\opt$ such
    that ${\tt in}(v_{i+1}^j)$ is contained in $S_{i+1}^{\ell}$. Thus
    by considering $S_{i}$s in a decreasing order and using induction,
    if $S_{i}^j$ is in $\opt$ then $v_{i+1}^j$ is reachable form
    $v_{p+1}^1$.

    Next consider a set $S_{p+i}^j$ that is selected in $\opt$
    ($1\leq i\leq p$). By similar argument, $T_{i}^j$ is not in $\opt$
    and there exists a set $S_{p+i-1}^\ell$ (or $S_1^\ell$ if $i=1$)
    in $\opt$ such that ${\tt out}(u_{i}^j)$ is contained in
    $S_{p+i-1}^\ell$.  Let
    $u_{i+1}^{\ell_1}, \cdots, u_{i+1}^{\ell_k}$ be the set of
    vertices whose corresponding ${\tt out}$ elements are in
    $S_{p+i}^j$. Then by induction, there exists an index $r$ such
    that $v_1^r$ is reachable from $v_{p+1}^1$ and $u_{1}^r$ is also
    reachable from all $u_{i+1}^{\ell_1}, \cdots, u_{i+1}^{\ell_k}$.
    Moreover, the way we constructed the instance $\mathsf{SC}$
    guarantees that all sets $S_{2p}^1, \cdots, S_{2p}^n$ contains
    ${\tt out}(u_{p+1}^1)$. Hence if the size of an optimal solution
    of $\mathsf{SC}$ is $(2p+1)n +1$ then the solution of
    $\mathsf{ISC}$ is $1$.
\end{proof}
\begin{corollary}%
    \corlab{cor:equivalence}%
    \prob{Intersection Set Chasing($n,p$)} returns $1$ if and only if
    the size of optimal solution of its corresponding \prob{Set Cover}
    instance (as described here) is $(2p+1)n +1$.
\end{corollary}


\begin{observation}%
    \obslab{obr:stream-communication-set-cover}%
    Any streaming algorithm for \probSetCover, $\sI$,
    that in $\ell$ passes solves the problem optimally with a probability of error
    ${\tt err}$ and consumes $s$ memory space, solves the 
    corresponding communication \probSetCover problem in $\ell$
    rounds using $O(s\ell^2)$ bits of communication with probability
    error ${\tt err}$.
\end{observation}
\begin{proof}
    Starting from player $P_1$, each player runs $\sI$ over its input
    sets and once $P_i$ is done with its input, she sends the working
    memory of $\sI$ publicly to other players. Then next player starts
    the same routine using the state of the working memory received
    from the previous player. Since $\sI$ solves the \probSetCover
    instance optimally after $\ell$ passes using $O(s)$ space with
    probability error ${\tt err}$, applying $\sI$ as a black box we
    can solve $\mathsf{P}$ in $\ell$ rounds using $O(s\ell^2)$ bits of
    communication with probability error ${\tt err}$.
\end{proof}

\begin{proofof} {\thmref{reduction}:}
    By \obsref{obr:stream-communication-set-cover}, any $\ell$-round
    $O(s)$-space algorithm that solves streaming \prob{Set
       Cover}($\GSet, \sS$) optimally can be used to solve the
    communication \prob{Set Cover($\GSet, \sS, p$)} problem in $\ell$
    rounds using $O(s\ell^2)$ bits of communication. Moreover, by
    \corref{cor:equivalence}, we can decide the solution of the
    communication \prob{Intersection Set Chasing($n, p$)} by solving
    its corresponding communication \prob{Set Cover problem}. Note
    that while working with the corresponding \probSetCover instance
    of \prob{Intersection Set Chasing($n,p$)}, all players know the
    collection of elements $\GSet$ and each player can construct its
    collection of sets $\sS_i$ using $f_i$ (or $f'_i$).

    However, by a result of \cite{go-slbmg-13}, we know that any
    protocol that solves the communication \prob{Intersection Set
       Chasing($n, p$)} problem with probability of error less than
    $1/10$, requires
    $\Omega({n^{1+1/(2p)} \over p^{16} \log^{3/2}n})$ bits of
    communication.  Since in the corresponding \probSetCover
    instance of the communication \prob{Intersection Set
       Chasing($n,p$)}, $\card{\GSet} = (2p+1)\times 2n + 2p = O(np)$
    and $\card{\sS} \leq (2p+1)n +2pn = O(np)$, any $(p-1$)-pass
    streaming algorithm that solves the \probSetCover problem
    optimally with a probability of error at most $1/10$, requires
    $\Omega({n^{1+1/(2p)} \over p^{18} \log^{3/2}n})$ bits of
    communication. Then using
    \obsref{obr:stream-communication-set-cover}, since
    $\delta \geq {\log \log n \over \log n}$, any
    $({1 \over 2\delta} -1)$-pass streaming
    algorithm
    of \probSetCover that finds an optimal solution with error
    probability less than $1/10$, requires
    $\tldOmega(\card{\sS} \cdot \card{\GSet}^{\delta})$ space.
\end{proofof}

\section{Lower Bound%
   for Sparse Set Cover in Multiple Passes}%
\seclab{sparse-lower-bound}

In this part we give a stronger lower bound for the instances of
the streaming \probSetCover problem with sparse input sets. An 
instance of the \probSetCover
problem is \prob{$s$-Sparse Set Cover}, if for each set $\range \in \sS$ we
have $\card{\range} \leq s$.  We can us the same reduction approach
described earlier in \secref{lower-bound-multipass} to show that any
$({1/2\delta}-1)$-pass streaming algorithm for \prob{$s$-Sparse
   Set Cover} requires $\Omega(\card{\sS}s)$ memory space if
$s < \card{\GSet}^{\delta}$ and $\sS = O(\GSet)$. To prove this, we need to explain more
details of the approach of \cite{go-slbmg-13} on the lower bound of
the communication \prob{Intersection Set Chasing} problem. They first
obtained a lower bound for \prob{Equal Pointer Chasing($n,p$)} problem
in which two instances of the communication \prob{Pointer
   Chasing($n,p$)} are given and the goal is to decide whether these
two instances point to a same value or not;
$f_p(\cdots f_1(1)\cdots) = f'_p(\cdots f'_1(1)\cdots)$.
\begin{definition}[$r$-non-injective functions]
    A function $f:[n] \rightarrow [n]$ is called
    $r$-\emph{non}-\emph{injective} if there exists $A\subseteq [n]$
    of size at least $r$ and $b\in [n]$ such that for all $a\in A$,
    $f(a)=b$.
\end{definition}
\begin{definition}[\prob{Pointer Chasing} Problem] 
	\prob{Pointer Chasing($n,p$)} is a $(p,$ $p-1)$
    communication problem in which the player $i$ has a function
    $f_i: [n] \rightarrow[n]$ and the goal is to compute
    $f_1(f_2(\cdots
    f_p(1)\cdots))$. 
\end{definition}
\begin{definition}[\prob{Equal Limited Pointer
       Chasing} Problem] \prob{Equal Pointer
       Chasing($n,p$)} is a $(2p,p-1)$ communication problem in which
    the first $p$ players have an instance of the \prob{Pointer
       Chasing($n,p$)} problem and the other $p$ players have another
    instance of the \prob{Pointer Chasing($n,p$)} problem. The output
    of the \prob{Equal Pointer Chasing($n,p$)} is $1$ if the solutions
    of the two instances of \prob{Pointer Chasing($n,p$)} have the
    same value and $0$ otherwise. Furthermore in another variant of pointer chasing problem, 
    \prob{Equal Limited
       Pointer Chasing($n,p,r$)}, if there exists $r$-non-injective
    function $f_i$, then the output is $1$. Otherwise, the output is
    the same as the value in \prob{Equal Pointer Chasing($n,p$)}.
\end{definition}
For a boolean communication problem \prob{P}, ${\sf OR}_t$(\prob{P})
is defined to be ${\sf OR}$ of $t$ instances of \prob{P} and the
output of ${\sf OR}_t$(\prob{P}) is ${\tt true}$ if and only if the
output of any of the $t$ instances is ${\tt true}$. Using a direct sum
argument, \cite{go-slbmg-13} showed that the communication complexity
of ${\sf OR}_t$(\prob{Equal Limited Pointer Chasing($n,p,r$)}) is $t$
times the communication complexity of \prob{Equal Limited Pointer
   Chasing($n,p,r$)}. 
\begin{lemma}[\cite{go-slbmg-13}]%
    \lemlab{elpc-lower-bound}%
    Let $n$, $p$, $t$ and $r$ be positive integers such that
    $n \geq 5p$, $t\leq {n\over 4}$ and $r=O(\log n)$. Then the amount
    of bits of communication to solve ${\sf OR}_t$(\prob{Equal Limited
       Pointer Chasing($n,p,r$)}) with error probability less than
    $1/3$ is $\Omega({tn \over p^{16}\log n}) - O(pt^2)$.
\end{lemma}
\begin{lemma}[\cite{go-slbmg-13}]%
    \lemlab{elpc-reduction-isc}%
    Let $n$, $p$, $t$ and $r$ be positive integers such that
    $t^{2p}r^{p-1} < {n/10}$. Then if there is a protocol that
    solves \prob{Intersection Set Chasing($n,p$)} with probability of
    error less than $1/10$ using $C$ bits of communication, there is a
    protocol that solves ${\sf OR}_t$(\prob{Equal Limited Pointer
       Chasing($n,p,r$)}) with probability of error at most $2/10$
    using $C+2p$ bits of communication.
\end{lemma}
 
\noindent 
Consider an instance of ${\sf OR}_t$(\prob{Equal Limited Pointer
   Chasing}($n,p,r$)) in which $t\leq n^{\delta}, r=\log(n), p={1\over 2\delta} -1$ where 
   ${1\over \delta} = o(\log n)$. By \lemref{elpc-lower-bound}, the required amount of bits of 
   communication to solve the instance with constant success probability is $\tldOmega(tn)$. 
   Then,applying \lemref{elpc-reduction-isc}, to solve the corresponding 
   \prob{Intersection Set Chasing}, $\tldOmega(tn)$ bits of communication
   is required.
    
In the reduction from ${\sf OR}_t$(\prob{Equal Limited Pointer Chasing($n,p,r$)}) to \prob{Intersection Set Chasing($n,p$)} (proof
of \lemref{elpc-reduction-isc}), the $r$-non-injective property is
preserved. In other words, in the corresponding \prob{Intersection Set
   Chasing} instance each player's functions
$f_i:[n] \rightarrow 2^{[n]}$ is union of $t$ $r$-non-injective
functions $f_i(a):=f_{i,1}(a)\cup \cdots \cup f_{i,t}(a)$\footnote{The
   \prob{Intersection Set Chasing} instance is obtained by overlaying
   the $t$ instances of \prob{Equal Pointer Chasing($n,p,r$)}. To be
   more precise, the function of player $i$ in instance $j$ is
   $\pi_{i,j} \circ f_{i,j} \circ \pi^{-1}_{i+1,j}$ ($\pi$ are
   randomly chosen permutation functions) and then stack the functions
   on top of each other.}.  Given that none of the $f_{i,j}$ functions
   is $r$-non-injective, the corresponding
\probSetCover instance will have sets of size at
most $rt$ ($S$-type sets are of size at most $t$ for $1\leq i\leq p$
and of size at most $rt$ for $p+1\leq i\leq 2p$).  Since
$r = O(\log n)$, the corresponding \probSetCover instance is 
$\tldO(t)$-sparse. As we showed earlier in the
reduction from \prob{Intersection Set Chasing} to \probSetCover, the
number of elements (and sets) in the corresponding \probSetCover instance is $O(np)$. 
Thus we have the following result for \prob{$s$-Sparse Set Cover} problem.
\begin{theorem} %
    \thmlab{sparse-set-cover}%
    For $s \leq \card{\GSet}^{\delta}$, any streaming algorithm that
    solves \prob{$s$-Sparse Set Cover($\GSet,\sS$)} optimally with
    probability of error less than $1/10$ in $({1 \over 2\delta}-1)$
    passes requires $\tldOmega(\card{\sS}s)$ memory space for $\sS = O(\GSet)$.
\end{theorem}

\hypersetup{linkcolor=black}%
\hypersetup{urlcolor=black}%
\bibliographystyle{alpha}%
{
   \bibliography{set_cover}%
}

\end{document}